\documentclass[preprint,12pt,authoryear]{elsarticle}

\usepackage{a4wide}
\usepackage{amssymb}
\usepackage{amsfonts}

\usepackage{amsmath,amsthm,amsfonts,amssymb,color}

\usepackage{multicol}

\usepackage{natbib}

\usepackage{multirow} 
\usepackage{float}
\usepackage{graphicx}
\usepackage{longtable}

\DeclareMathOperator\argmin{argmin}

\usepackage{amssymb}
\usepackage{amsthm}
\newtheorem{theorem}{Theorem}
\newtheorem{lemma}{Lemma}

\newtheorem{definition}{Definition}
\usepackage[authoryear]{natbib}

\begin{document}

\begin{frontmatter}

\title{Two-phase source and reaction coefficient Stefan type problems} 


\author{Targyn A. Nauryz$^{1,2,3}$}

\address{$^1$International School of Economics, Kazakh-British Technical University, Almaty, Kazakhstan,\\
$^2$Institute of Mathematics and Mathematical Modeling, Almaty, Kazakhstan,\\
$^3$School of Digital Technologies, Narxoz University, Almaty, Kazakhstan\\
Email: targyn.nauryz@gmail.com}

\begin{abstract}
This paper investigates inverse two-phase Stefan-type problems for parabolic heat equations with unknown time-dependent source and reaction coefficients. Suitable transformations reduce the original free-boundary problems to equivalent equations on fixed spatial domains, and Fourier spectral expansions are used to derive reconstruction formulas for the unknown coefficients. In the first formulation, the source coefficients are identified from integral, pointwise, and nonlocal additional data, leading to Volterra integral equations. In the second formulation, an exponential transformation is applied to recover the reaction coefficients. A principal advantage of the model with two moving boundaries is that each boundary provides its own Stefan condition which sufficient to determine the two unknown time-dependent coefficients without additional overdetermination conditions. Under appropriate assumptions the existence, uniqueness, boundedness and regularity of weak and strong solutions are established. Illustrative examples confirm the applicability of the reconstruction procedure and show that the coefficients remain stable under the noisy data. The proposed approach provides a rigorous base for identifying time-dependent thermal parameters in two-phase phase-change processes.
\end{abstract}


\end{frontmatter}

\section{Introduction}

Free-boundary problems for parabolic equations arise naturally in the mathematical description of phase transitions, combustion, melting, solidification, thermal energy storage, and other processes in which the spatial region occupied by a phase changes with time. In such models, the temperature field is governed by a heat equation in each phase, while the motion of the interface is determined by an energy-balance condition involving the latent heat. Classical investigations of Stefan-type and
free-boundary heat problems established the mathematical foundations for the analysis of temperature distributions and moving interfaces under different boundary and initial conditions \cite{2}, \cite{5}. Evolution equations posed in domains with moving boundaries have also been studied for other physical systems, including elastic beams and strings with moving endpoints \cite{3}, \cite{4}. General analytical tools for parabolic partial differential equations and their weak formulations can be found in \cite{16}.

In many practical applications, some physical characteristics of the thermal process cannot be measured directly and must be recovered from indirect observations. This leads to inverse problems in which an unknown source intensity, reaction coefficient, boundary flux, or control parameter is reconstructed simultaneously with the state variable. Inverse coefficient problems for parabolic equations have been considered using additional integral data, singla data observations, boundary measurements, and
nonlocal conditions \cite{7}, \cite{8}, \cite{9}, \cite{11}, \cite{18}. General theoretical and computational aspects of inverse problems for parabolic equations and their applications are discussed in \cite{12}, \cite{14}. Related developments include multidimensional inverse Cauchy problems for evolution equations \cite{13}, inverse source identification for subelliptic heat equations \cite{15}, and inverse source problems for fractional diffusion equations with nonlocal boundary conditions \cite{19}.

Spectral methods provide an effective framework for constructing solutions of inverse parabolic problems. After expanding the solution with respect to the eigenfunctions of an associated elliptic operator, the original partial differential equation can be reduced to a system of ordinary differential equations for the Fourier coefficients. The convergence and regularity of the resulting series depend on suitable estimates for the eigenvalues and eigenfunctions of the corresponding spectral problem \cite{1}, \cite{10}. Spectral techniques are also widely used in inverse and imaging problems
in which the qualitative behavior of eigenstates plays an essential role \cite{6}. Together with energy estimates, Bessel-type inequalities, and Gronwall-type arguments \cite{17}, this approach allows one to derive reconstruction formulas and establish the existence, uniqueness, and stability of the recovered coefficients.

Inverse Stefan problems are more complicated than inverse problems posed on fixed domains because the unknown temperature fields are coupled through a moving interface. The geometry of each phase changes with time, the transformed equations contain variable coefficients and transport terms, and the interface conditions connect the temperature gradients in both phases. Moreover, the recovery of an unknown time-dependent coefficient generally requires additional information ensuring that the inverse problem is not underdetermined. Such information may be prescribed in the form of an integral overdetermination condition, a temperature measurement at an interior point, a nonlocal observation, or a boundary heat-flux relation.

A recent study considered inverse Stefan problems for determining a time-dependent source coefficient and an unknown heat-flux function in a one-phase setting \cite{20}. The present work extends the approach of \cite{20} to a genuinely two-phase Stefan problem. This extension is nontrivial because two temperature fields must be determined in two time-dependent subdomains, and the unknown temporal coefficients in the two phases are coupled through the temperature continuity and Stefan flux-balance conditions at the moving interface. Consequently, the present analysis requires separate transformations of both moving subdomains, two spectral
representations, coupled reconstruction relations, and additional estimates  ensuring the convergence and regularity of the solutions in both phases.

Two inverse two-phase Stefan-type formulations are investigated. In the first
formulation, the material occupies a fixed one-dimensional body divided by a moving interface into two phases. The governing heat equations contain unknown time-dependent source coefficients \(R_1(t)\) and \(R_2(t)\). Nonhomogeneous boundary conditions are first reduced to homogeneous ones, and each moving subdomain is transformed onto the fixed interval \((0,1)\). The transformed temperature functions are then represented by Fourier series associated with the Dirichlet Sturm--Liouville problem. To recover the source coefficients, three different types of additional information are considered: an integral overdetermination condition, a single interior observation, and a nonlocal integral condition. In each case, the reconstruction relation is reduced to a Volterra integral equation.

The second formulation concerns a two-phase Stefan-type problem with Neumann boundary conditions and two moving boundaries. The governing equations contain unknown reaction-type coefficients \(P_1(t)\) and \(P_2(t)\). An exponential substitution is used to remove these coefficients from the differential equations, after which the transformed problems are solved on fixed spatial intervals. The unknown coefficients are subsequently reconstructed from the corresponding temporal factors. An important feature of this formulation is that the Stefan conditions imposed at the moving
boundaries contain sufficient information to determine the unknown coefficients. Therefore, no separate integral, pointwise, or nonlocal overdetermination condition is required for their recovery.

The main contributions of the present paper can be summarized as follows. First, one-phase inverse Stefan analysis developed in \cite{20} is extended to a coupled two-phase configuration. Second, suitable changes of variables are introduced to transform the moving-boundary problems into equivalent parabolic systems on fixed domains. Third, spectral representations of the temperature fields and explicit reconstruction relations for the unknown source and reaction coefficients are derived. Fourth, the resulting inverse relations are formulated as Volterra integral equations, and sufficient regularity, compatibility and nondegeneracy conditions are established for their solvability. Fifth, the absolute and uniform convergence of the spectral series is justified using integration by parts, Cauchy-Schwarz and Bessel inequalities, and Gronwall-type estimates. These results lead to the existence and uniqueness of weak and strong solutions for the transformed problems and, consequently, for the original moving-boundary systems. Finally, illustrative examples and numerical experiments are presented to demonstrate the reconstruction of the unknown coefficients and to examine their behavior under noisy data.

The structure of the paper is organized as follows. Section~2 introduces the first inverse two-phase Stefan problem and transforms it from moving subdomains to fixed spatial intervals. Section~3 develops the spectral representation and studies the recovery of the source coefficients from integral, single and nonlocal additional data. The existence, uniqueness, convergence, and regularity properties of the solutions are also established. Numerical and analytical examples for the first formulation are presented in the subsequent sections. The second inverse Stefan-type formulation with two moving boundaries and unknown reaction coefficients is then investigated, followed by numerical illustrations of the coefficient
reconstruction. The final section summarizes the principal conclusions and possible directions for further research.

\section{The source coefficient Stefan type problem with one moving boundary}
We consider the domains $\Omega_1=(0,s(t))\times(0,T)$ and $\Omega_2=(s(t),d)\times(0,T)$ where $T>0$ is given arbitrary fixed time interest. An objective of the work is to determine weak solution pairs $(u_1, R_1)$ and $(u_2,R_2)$ of the following thermal system:
\begin{equation}\label{eq1}
    \begin{cases}
        \dfrac{\partial u_1}{\partial t}-a_1^2\dfrac{\partial^2 u_1}{\partial x^2}=R_1(t)f_1(x,t),& (x,t)\in \Omega_1\\
        \dfrac{\partial u_2}{\partial t}-a_2^2\dfrac{\partial^2 u_2}{\partial x^2}=R_2(t)f_2(x,t),& (x,t)\in \Omega_2,\\
        u_1(x,0)=\varphi_1(x),\quad s(0)=s_0,& \forall x\in[0,s(0)],\\
        u_2(x,0)=\varphi_2(x),\quad s(0)=s_0,& \forall x\in[s(0),d],\\
        u_1(0,t)=g_1(t),&\forall t\in[0,T],\\
        u_1(s(t),t)=u^*,& \forall t\in[0,T],\\
        u_2(s(t),t)=u^*,& \forall t\in[0,T],\\
        -k_1\dfrac{\partial u_1}{\partial x}\bigg|_{x=s(t)}=-k_2\dfrac{\partial u_2}{\partial x}\bigg|_{x=s(t)}+L s'(t),&\forall t\in[0,T],\\
        u_2(d,t)=g_2(t),&\forall t\in[0,T].
    \end{cases}
\end{equation}
where $a_1^2,a_2^2>0$ are thermal diffusivity for each phase, $\varphi_i(x), g_i(t), i=1,2$ are given positive functions, $s_0, d$ are given positive constants, $u^*$ is a phase change temperature, $s(t)$ is a moving boundary describing phase change interface and $L$ is a latent heat of fusion, $k_i, i=1,2$ is a thermal conductivity of each phase.

Introducing a transformations
\begin{equation}\label{eq2}
    v_1(x,t)=u_1(x,t)-g_1(t)-\dfrac{u^*-g_1(t)}{s(t)}x,\quad (x,t)\in\Omega_1,
\end{equation}
\begin{equation}\label{eq3}
    v_2(x,t)=u_2(x,t)-u^*-\dfrac{g_2(t)-u^*}{d-s(t)}(x-s(t)),\quad (x,t)\in\Omega_2,
\end{equation}
the problem \eqref{eq1} can be represented as following form:
\begin{equation}\label{eq4}
    \begin{cases}
        \dfrac{\partial v_1}{\partial t}-a_1^2\dfrac{\partial^2v_1}{\partial x^2}=R_1(t)h_1(x,t),&(x,t)\in\Omega_1,\\
        \dfrac{\partial v_2}{\partial t}-a_2^2\dfrac{\partial^2v_2}{\partial x^2}=R_2(t)h_2(x,t),&(x,t)\in\Omega_2,\\
        v_1(x,t)=\psi_1(x),\quad s(0)=s_0,& \forall x\in[0,s(0)],\\
        v_2(x,t)=\psi_2(x),\quad s(0)=s_0,& \forall x\in[s(0),d],\\
        v_1(0,t)=v_1(s(t),t)=0,& \forall t\in[0,T],\\
        v_2(s(t),t)=v_2(d,t)=0,& \forall t\in[0,T],\\
        -k_1\dfrac{\partial v_1}{\partial x}\bigg|_{x=s(t)}=-k_2\dfrac{\partial v_2}{\partial x}\bigg|_{x=s(t)}+\chi(t)+Ls'(t),& \forall t\in[0,T],
    \end{cases}
\end{equation}
where $\chi(t)=k_1\frac{u^*-g_1(t)}{s(t)}-k_2\frac{g_2(t)-u^*}{d-s(t)}$ and
\begin{equation}\label{eq5}
    h_1(x,t)=f_1(x,t)-\dfrac{g_1(t)}{R_1(t)}+\dfrac{g_1'(t)s(t)+s'(t)(u^*-g_1(t))}{s^2(t)R_1(t)},
\end{equation}
\begin{equation}\label{eq6}
    h_2(x,t)=f_2(x,t)-\dfrac{g_2'(t)(d-s(t))+s'(t)(g_2(t)-u^*)}{R_2(t)(d-s(t))^2}(x-s(t))+\dfrac{g_2(t)-u^*}{R_2(t)(d-s(t))}s(t),
\end{equation}
\begin{equation}\label{eq7}
    \psi_1(x)=\varphi_1(x)-g_1(0)-\dfrac{u^*-g_1(0)}{s(0)}x,\quad \psi_2(x)=\varphi_2(x)-u^*-\dfrac{g_2(0)-u^*}{d-s(0)}(x-s(0)).
\end{equation}
Considering the change of variable for transformation the domains $\Omega_1$ and $\Omega_2$ with moving boundary to the fixed domain $\Omega=(0,1)\times(0,T)$, that is
\begin{equation}\label{eq8}
    v_1(x,t)=V_1(\xi_1,t),\quad\quad V_1: \Omega_1\longrightarrow\Omega,\quad (x,t)\mapsto (\xi_1,t)=\left(\dfrac{x}{s(t)},t\right),
\end{equation}
\begin{equation}\label{eq9}
    v_2(x,t)=V_2(\xi_2,t),\quad\quad V_2: \Omega_2\longrightarrow\Omega,\quad (x,t)\mapsto(\xi_2,t)=\left(\dfrac{x-s(t)}{d-s(t)},t\right).
\end{equation}
We need to note that $V_1\in C^{2,1}(\Omega)$ and $V_2\in C^{2,1}(\Omega)$, then by inverse function definition $V_1^{-1}\in C^{2,1}(\Omega_1)$ and $V_2^{-1}\in C^{2,1}(\Omega_2)$. 

Applying transformation functions \eqref{eq8} and \eqref{eq9}, the problem \eqref{eq4} reduced to
\begin{equation}\label{eq10}
    \begin{cases}
        \dfrac{\partial V_1}{\partial t}-b_1(\xi_1,t)\dfrac{\partial V_1}{\partial \xi_1}-a_1(t)\dfrac{\partial^2 V_1}{\partial\xi_1^2}=R_1(t)F_1(\xi_1,t),& (\xi_1,t)\in\Omega,\\
        \dfrac{\partial V_2}{\partial t}-b_2(\xi_2,t)\dfrac{\partial V_2}{\partial \xi_2}-a_2(t)\dfrac{\partial^2 V_2}{\partial\xi_2^2}=R_2(t)F_2(\xi_2,t),& (\xi_2,t)\in\Omega,\\
        V_1(\xi_1,0)=\Phi_1(\xi_1),&\forall\xi_1\in[0,1],\\
        V_2(\xi_2,0)=\Phi_2(\xi_2),&\forall\xi_2\in[0,1],\\
        V_1(0,t)=V_1(1,t)=0,&\forall t\in[0,T],\\
        V_2(0,t)=V_2(1,t)=0,&\forall t\in[0,T],\\
        -\dfrac{k_1}{s(t)}\dfrac{\partial V_1}{\partial \xi_1}\bigg|_{\xi_1=1}=-\dfrac{k_2}{d-s(t)}\dfrac{\partial V_2}{\partial\xi_2}\bigg|_{\xi_2=0}+\chi(t)+Ls'(t),& \forall t\in[0,T].
    \end{cases}
\end{equation}
where $F_1(\xi_1,t)=h_1(\xi_1s(t),t)$, $F_2(\xi_2,t)=h_2(\xi_2(d-s(t))+s(t),t)$, $\Phi_1(\xi_1)=\psi_1(\xi_1s(0))$, $\Phi_2(\xi_2)=\psi_2(\xi_2(d-s(0))+s(0))$, $b_1(\xi_1,t)=\frac{\xi_1s'(t)}{s(t)}$, $b_2(\xi_2,t)=\frac{s'(t)(\xi_2-1)}{d-s(t)}$, $a_1(t)=\frac{a_1^2}{s^2(t)}$ and $a_2(t)=\frac{a_2^2}{(d-s(t))^2}$. By definition of inverse functions we have
$$F_i(\xi_i,t)=h_i(V_i^{-1}(\xi_i,t))=h_i(x,t),\quad i=1,2.$$
$$V_i(\xi_i,t)=v_i(V_i^{-1}(\xi_i,t))=v_i(x,t),\quad i=1,2.$$
A primary objective of our study is to present weak solution of the problem \eqref{eq10} with given single observation point, overdetermined integral condition and non-local data for the first phase of body. We need to verify existence and uniqueness of the problem \eqref{eq10} which consequently implies the existence and uniqueness of the problem \eqref{eq4} and \eqref{eq1}. Additionally, we are interested to show existence and uniqueness of the global and weak solutions of the problem \eqref{eq10}.

\section{Weak formulation of the solution}
For convenience, we have also used  
$((\cdot))$, $(\cdot)$  to denote the scalar product and $||\cdot||$, $|\cdot|$ the norms in $\mathbb{H}_0^1(0,1)$ and $L^2(0,1)$, respectively.

Introducing a spectral problem to the problem \eqref{eq10}:
\begin{equation}\label{eq11}
    \begin{cases}
        -\phi_i''(\xi_i)=\lambda_i\phi_i(\xi_i),&\xi_i\in(0,1),\\
        \phi_i(0)=\phi_i(1)=0,
    \end{cases}
\end{equation}
where $i=1,2$ for each phases. Thus, \eqref{eq11} gives us the eigenvalue and eigenfunction
\begin{equation}\label{eq12}
    \phi_{i,n}(\xi)=\sqrt{2}\sin\left(\sqrt{\lambda_{n}}\xi_i\right),\quad \lambda_{n}=n^2\pi^2,\quad n\in\mathbb{N},\quad i=1,2.
\end{equation}
\begin{definition}
A function $V_i\in C([0,T],L^2(0,1))\cap C((0,T];\mathbb{H}_0^2(0,1))\cap C^1((0,T];L^2(0,1))$, where $i=1,2$ is called weak solution of the problem \eqref{eq10} if it satisfies
\begin{equation}\label{eq13}
    \left(\dfrac{\partial V_i}{\partial t},\phi_{i,n}\right)_{L^2(0,1)}-\left(b_i(\xi_i,t)\dfrac{\partial V_i}{\partial \xi_i},\phi_{i,n}\right)_{L^2(0,1)}-a_i(t)\left(\dfrac{\partial^2 V_i}{\partial\xi_i^2},\phi_{i,n}\right)_{L^2(0,1)}=R_i(t)\left(F_i(\xi_i,t),\phi_{i,n}\right)_{L^2(0,1)}
\end{equation}
with initial condition
\begin{equation}\label{eq14}
    \left(V_i(\xi_i,0),\phi_{i,n}\right)_{L^2(0,1)}=\left(\Phi_i(\xi_i),\phi_{i,n}\right),\quad \forall\xi_i\in[0,1]
\end{equation}
where $n\in\mathbb{N}$, $\phi_{i,n}$ is defined by \eqref{eq11} and $b_1(\xi_1,t)=\frac{\xi_1s'(t)}{s(t)}$, $b_2(\xi_2,t)=\frac{(1-\xi_2)s'(t)}{d-s(t)}$, $a_1(t)=\frac{a_1^2}{s^2(t)}$, $a_2(t)=\frac{a_2^2}{(d-s(t))^2}$. 
\end{definition}

Consequently we have the next definition.
\begin{definition}
    The functions $u_1\in C([0,T],L^2(0,s(t)))\cap C((0,T];\mathbb{H}_0^2(0,s(t)))\cap C^1((0,T];L^2(0,s(t)))$ and $u_2\in C([0,T],L^2(s(t),d))\cap C((0,T];\mathbb{H}_0^2(s(t),d))\cap C^1((0,T];L^2(s(t),d))$  are said to be weak solution of the problem \eqref{eq1} such that the change of variable functions $v_1\left(x,t\right)=V_1\left(\frac{x}{s(t)},x\right)$ and $v_2\left(x,t\right)=V_2\left(\frac{x-s(t)}{d-s(t)},t\right)$ implies temperature solutions of the problem \eqref{eq1} by
$$u_1(x,t)=V_1\left(\frac{x}{s(t)},t\right)+g_1(t)+\dfrac{u^*-g_1(t)}{s(t)}x,\quad \forall(x,t)\in\Omega_1,$$
$$u_2(x,t)=V_2\left(\frac{x-s(t)}{d-s(t)},t\right)+u^*+\dfrac{g_2(t)-u^*}{d-s(t)}(x-s(t)),\quad \forall(x,t)\in\Omega_2,$$
using \eqref{eq2} and \eqref{eq3}, 
\end{definition}

\subsection{Integral overdetermination condition.}\label{sec1} Let us now show the existence and uniqueness of a pair of functions $(u_1,R_1)$ and $(u_2,R_2)$ such that $R(t)\in C[0,T]$, and $u_1$ and $u_2$ are weak solution the problem \eqref{eq1} satisfying conditions
    \begin{equation}\label{eq15}
        \int\limits_0^{s(t)}u_1(x,t)dx=\omega(t),\quad t\in[0,T],
    \end{equation}
    \begin{equation}\label{eq16}
        -k_1\dfrac{\partial u_1}{\partial x}\bigg|_{x=s(t)}=-k_2\dfrac{\partial u_2}{\partial x}\bigg|_{x=s(t)}+L s'(t),\quad t\in[0,T],
    \end{equation}
where $\omega(t)$ and $s(t)$ are given function.

Consequently, for the problem \eqref{eq10} the conditions \eqref{eq15} and \eqref{eq16} takes the form:
\begin{equation}\label{eq17}
    \int\limits_0^1 V_1(\xi_1,t)d\xi_1=\dfrac{\omega(t)}{s(t)}-g_1(t)+\dfrac{u^*-g_1(t)}{2},\quad t\in[0,T],
\end{equation}
\begin{equation}\label{eq18}
    -\dfrac{k_1}{s(t)}\dfrac{\partial V_1}{\partial \xi_1}\bigg|_{\xi_1=1}=-\dfrac{k_2}{d-s(t)}\dfrac{\partial V_2}{\partial\xi_2}\bigg|_{\xi_2=0}+k_1\frac{u^*-g_1(t)}{s(t)}-k_2\frac{g_2(t)-u^*}{d-s(t)}+Ls'(t),\quad t\in[0,T],
\end{equation}
in which the solution pairs $(V_1,R_1)$, $(V_2, R_2)$ have to be determined.

Let $\left\{\phi_{i,n}(\xi_i)\right\}_{n=1}^{\infty}$ be orthonormal eigenfunction set of the spectral problem \eqref{eq11}. The weak solution of the problem \eqref{eq10} is sought in the form 
\begin{equation}\label{eq19}
    V_i(\xi_i,t)=\sum\limits_{n=1}^{\infty}V_{i,n}(t)\phi_{i,n}(\xi_i),
\end{equation}
where $(V_i,\phi_{i,n})_{L^2(0,1)}=V_{i,n}(t)$ for $i=1,2$ and $n\in\mathbb{N}$.

Substituting \eqref{eq19} into equation \eqref{eq13} and using definition of $\phi_{i,n}$ defined by \eqref{eq12}, and introducing equalities
$$\left(\dfrac{\partial V_i}{\partial t},\phi_{i,n}\right)_{L^2(0,1)}=\dfrac{d}{dt}\left(V_{i,n},\phi_{i,n}\right)_{L^2(0,1)}=V_{i,n}'(t),\quad i=1,2,\quad n\in\mathbb{N},$$
$$\left(V_i,\phi_{i,n}\right)_{\mathbb{H}_0^2(0,1)}=-\lambda_{i,n}\left(V_i,\phi_{i,n}\right)_{L^2(0,1)}=-\lambda_{n}V_{i,n}(t),\quad i=1,2, \quad n\in\mathbb{N}$$
$$(F_{i},\phi_{i,n})_{L^2(0,1)}=F_{i,n}(t),\quad i=1,2,\quad n\in\mathbb{N},$$
then we obtain 
\begin{equation}\label{eq20}
    \begin{cases}
        V_{1,n}'(t)+(a_1(t)\lambda_{n}-b_1(t))V_{1,n}(t)=R_1(t)F_{1,n}(t),& n\in\mathbb{N},\quad t\in[0,T],\\
    V_{2,n}'(t)+(a_2(t)\lambda_{n}-b_2(t))V_{2,n}(t)=R_2(t)F_{2,n}(t),& n\in\mathbb{N},\quad t\in[0,T],\\
    V_{1,n}(0)=\Phi_{1,n}, & n\in\mathbb{N},\\
    V_{2,n}(0)=\Phi_{2,n}, & n\in\mathbb{N},\\
    \end{cases}
\end{equation}
where $b_i(t)=\left(b_{i}(\xi_i,t)\phi_{i,n}',\phi_{i,n}(\xi_i)\right)_{L^2(0,1)}$ and $\Phi_{i,n}=\left(\Phi_i(\xi_i),\phi_{i,n}(\xi_i)\right)_{L^2(0,1)},\quad i=1,2$.

One can easily estimate that the solution of the problem \eqref{eq20} has the following form
\begin{equation}\label{eq21}
    V_{i,n}(t)=\Phi_{i,n}e^{-\int_0^t(a_i(\tau)\lambda_{n}-b_i(\tau))d\tau}+\int_0^tR_i(\tau)F_{i,n}(\tau)e^{-\int_{\tau}^t(a_i(z)\lambda_{n}-b_i(z))dz}d\tau,
\end{equation}
where $i=1,2$ and $n\in\mathbb{N}$, then setting \eqref{eq21} into equation \eqref{eq19}, we can verify that solution of the problem \eqref{eq10}
\begin{equation}\label{eq22}
    \begin{split}
        V_i(\xi_i,t)&=\sqrt{2}\sum\limits_{n=1}^{\infty}\Bigg(\Phi_{i,n}e^{-\int_0^t(a_i(\tau)\lambda_{n}-b_i(\tau))d\tau}\\
        &+\int_0^tR_i(\tau)F_{i,n}(\tau)e^{-\int_{\tau}^t(a_i(z)\lambda_{n}-b_i(z))dz}d\tau\Bigg)\sin\left(\sqrt{\lambda_{n}}\xi_i\right),\quad i=1,2.
    \end{split}
\end{equation}
which represents the time-dependent behavior of a physical system governed by a partial differential equation, expressed as a series expansion of eigenfunctions in space and time-dependent coefficients. 

Differentiating both sides of \eqref{eq17} respect to $t$ and applying definition \eqref{eq22}, we can present the equation
\begin{equation}\label{eq23}
    \begin{split}
        &\sqrt{2}\sum\limits_{n=1}^{\infty}\Bigg[(b_1(t)-a_1(t)\lambda_{n})\Phi_{1,n}e^{-\int_0^t(a_1(\tau)\lambda_{n}-b_1(\tau))d\tau}\\
        &+(b_1(t)-a_1(t)\lambda_{n})\int_0^tR_1(\tau)F_{1,n}e^{-\int_{\tau}^t(a_1(z)\lambda_{n}-b_1(z))dz}d\tau+R_1(t)F_{1,n}(t)\Bigg]\int_0^1\sin\left(\sqrt{\lambda_{n}}\xi_1\right)d\xi_1\\
        &=\dfrac{\omega'(t)s(t)-s'(t)\omega(t)}{s^2(t)}-g_1'(t)+\dfrac{g_1'(t)}{2},
    \end{split}
\end{equation}
taking into account
$\int_0^1\sin\left(\sqrt{\lambda_{n}}\xi_1\right)d\xi_1=\frac{2}{\sqrt{\lambda_{n}}}$ when $n$ is odd, otherwise integral equals to zero, then we can determine the solution of the function $R_1(t)$
\begin{equation}\label{eq24}
    \begin{split}
       R_1(t)&=\dfrac{1}{2\sqrt{2}\sum\limits_{n=1}^{\infty}\frac{F_{1,2n-1}}{\sqrt{\lambda_{2n-1}}}}\Bigg[\dfrac{\omega'(t)s(t)-s'(t)\omega(t)}{s^2(t)}-\dfrac{1}{2}g_1'(t)\\
       &+2\sqrt{2}a_1(t)\sum\limits_{n=1}^{\infty}\sqrt{\lambda_{2n-1}}\Phi_{1,2n-1}e^{-\int_0^t(a_1(\tau)\lambda_{2n-1}-b_1(\tau))d\tau}\\
       &-2\sqrt{2}b_1(t)\sum\limits_{n=1}^{\infty}\dfrac{1}{\sqrt{\lambda_{2n-1}}}\Phi_{1,2n-1}e^{-\int_0^t(a_1(\tau)\lambda_{2n-1}-b_1(\tau))d\tau}\\
       &+2\sqrt{2}a_1(t)\int_0^t\left(\sum\limits_{n=1}^{\infty}\sqrt{\lambda_{2n-1}}F_{1,2n-1}(\tau)e^{-\int_{}\tau^t(a_1(z)\lambda_{2n-1}-b_1(z))dz}\right)R_1(\tau)d\tau\\
       &-2\sqrt{2}b_1(t)\int_0^t\left(\sum\limits_{n=1}^{\infty}\dfrac{1}{\sqrt{\lambda_{2n-1}}}F_{1,2n-1}(\tau)e^{-\int_{\tau}^t(a_1(z)\lambda_{2n-1}-b_1(z))d\tau}\right)R_1(\tau)d\tau\Bigg].
    \end{split}
\end{equation}
The equation \eqref{eq18} can be represented in the form
\begin{equation}
    \begin{split}
        &\dfrac{\sqrt{2}k_2}{d-s(t)}\sum\limits_{n=1}^{\infty}\Bigg[\sqrt{\lambda_{n}}\Phi_{2,n}e^{-\int_0^t(a_2(\tau)\lambda_{n}-b_2(\tau))d\tau}+\int_0^t\sqrt{\lambda_{n}}R_2(\tau)F_{2,n}(\tau)e^{-\int_{\tau}^t(a_2(z)\lambda_{n}-b_2(z))dz}d\tau\Bigg]\\
        &=\dfrac{\sqrt{2}k_1}{s(t)}\sum\limits_{n=1}^{\infty}\Bigg[(-1)^n\sqrt{\lambda_{n}}\Phi_{1,n}e^{-\int_0^t(a_1(\tau)\lambda_{n}-b_1(\tau))d\tau}\\
        &+\int_0^t(-1)^n\sqrt{\lambda_{n}}R_1(\tau)F_{1,n}(\tau)e^{-\int_{\tau}^t(a_1(z)\lambda_{n}-b_1(z))dz}d\tau\Bigg]+k_1\dfrac{u^*-g_1(t)}{s(t)}-k_2\dfrac{g_2(t)-u^*}{d-s(t)}+Ls'(t),
    \end{split}  
\end{equation}
To determine $R_2(t)$, we need to carefully analyze the equation provided. The term $R_2(t)$ appears in an integral, so finding it requires isolating and solving for $R_2(t)$ in the context of the equation. Differentiating both sides of the equation with respect to time $t$ is a valid approach to isolate $R_2(t)$, especially because the integral terms involve dependencies on $t$. Then we can estimate the solution $R_2(t)$ that is

\begin{equation}\label{eq25}
    \begin{split}
        R_2(t)&=\dfrac{1}{\sum\limits_{n=1}^{\infty}\sqrt{\lambda_{n}}F_{2,n}(t)}\Bigg[a_2(t)\sum\limits_{n=1}^{\infty}\lambda_{2,n}\sqrt{\lambda_{n}}\Phi_{2,n}e^{-\int_0^t(a_2(\tau)\lambda_{n}-b_2(\tau))d\tau}\\
        &-\left(b_2(t)+\dfrac{s'(t)}{d-s(t)}\right)\sum\limits_{n=1}^{\infty}\sqrt{\lambda_{n}}\Phi_{2,n}e^{-\int_0^t(a_2(\tau)\lambda_{n}-b_2(\tau))d\tau}\\
        &+a_2(t)\int_0^t\left(\sum\limits_{n=1}^{\infty}\lambda_{n}\sqrt{\lambda_{n}}F_{2,n}(\tau)e^{-\int_{\tau}^t(a_2(z)\lambda_{n}-b_2(z))dz}\right)R_2(\tau)d\tau\\
        &-\left(b_2(t)+\dfrac{s'(t)}{d-s(t)}\right)\int_0^t\left(\sum\limits_{n=1}^{\infty}\sqrt{\lambda_{n}}F_{2,n}(\tau)e^{-\int_{\tau}^t(a_2(z)\lambda_{n}-b_2(z))dz}\right)R_2(\tau)d\tau\\
        &+\left(\dfrac{k_1(d-s(t))}{k_2s(t)}b_1(t)-\dfrac{k_1s'(t)(d-s(t))}{k_2s^2(t)}\right)\sum\limits_{n=1}^{\infty}(-1)^n\sqrt{\lambda_{n}}\Phi_{1,n}e^{-\int_0^t(a_1(\tau)\lambda_{n}-b_1(\tau))d\tau}\\
        &+\left(\dfrac{k_1(d-s(t))}{k_2s(t)}b_1(t)-\dfrac{k_1s'(t)(d-s(t))}{k_2s^2(t)}\right)\\
        &\times\int_0^t\left(\sum\limits_{n=1}^{\infty}(-1)^n\sqrt{\lambda_{n}}F_{1,n}(\tau)e^{-\int_{\tau}^t(a_1(z)\lambda_{n}-b_1(z))dz}\right)R_1(\tau)d\tau\\
        &-\dfrac{k_1(d-s(t))}{k_2s(t)}a_1(t)\sum\limits_{n=1}^{\infty}(-1)^n\lambda_{n}\sqrt{\lambda_{1,n}}\Phi_{1,n}e^{-\int_0^t(a_1(\tau)\lambda_{n}-b_1(\tau))d\tau}\\
        &-\dfrac{k_1(d-s(t))}{k_2s(t)}a_1(t)\int_0^t\left(\sum\limits_{n=1}^{\infty}(-1)^n\lambda_{n}\sqrt{\lambda_{n}}F_{1,n}(\tau)e^{-\int_{\tau}^t(a_1(z)\lambda_{n}-b_1(z))dz}\right)R_1(\tau)d\tau\\
        &+\dfrac{k_1(d-s(t))}{k_2s(t)}R_1(t)\sum\limits_{n=1}^{\infty}(-1)^n\sqrt{\lambda_{n}}F_{1,n}(t)-\dfrac{k_1(d-s(t))}{\sqrt{2}k_2s^2(t)}(g_1'(t)s(t)+s'(t)(u^{*}-g_1(t)))\\
        &-\dfrac{g_2'(t)(d-s(t))+s'(t)(g_2(t)-u^{*})}{\sqrt{2}(d-s(t))}+\dfrac{L(d-s(t))}{\sqrt{2}k_2}s''(t)\Bigg].
    \end{split}
\end{equation}
Applying integration by parts, Caushy-Schwarz and classical Bessel inequalities we can present the following preliminary results which help us to prove uniform convergence of \eqref{eq24} and \eqref{eq25}.
\begin{lemma}\label{lem1}
    If $\Phi_1\in C^4[0,1]$ and $F_1(\xi,t)\in C([0,T];L^2(0,1))$ such that $\Phi_{1}(0)=\Phi_{1}(1)=\Phi_{1}'(0)=\Phi_{1}'(1)=\Phi_{1}''(0)=\Phi_{1}''(1)=\Phi_{1}'''(0)=\Phi_{1}'''(1)=0$, $F_1(0,t)=F_1(1,t)=F_1'(0,t)=F_1'(1,t)=F_1''(0,t)=F_1''(1,t)=0$ and $0<s_m<s(t)<s_M$, where $s(t)\in C^2[0,T]$, then following assumptions hold for all $t\in[0,T]$ that are
    \begin{equation}\label{eq26}
        \sum\limits_{n=1}^{\infty}\sqrt{\lambda_{n}}\Phi_{1,n}e^{-\int_0^t(a_1(\tau)\lambda_{n}-b_1(\tau))d\tau}\leq c_1||\Phi_1||_{C^4[0,1]},
    \end{equation}
    \begin{equation}\label{eq27}
        \sum\limits_{n=1}^{\infty}\dfrac{1}{\sqrt{\lambda_{2n-1}}}\Phi_{1,2n-1}e^{-\int_0^t(a_1(\tau)\lambda_{2n-1}-b_1(\tau))d\tau}\leq c_2||\Phi_1||_{C^4[0,1]},
    \end{equation}
    \begin{equation}\label{eq28}
        \sum\limits_{n=1}^{\infty}\sqrt{\lambda_{n}}F_{1,n}(\tau)e^{-\int_{\tau}^t(a_1(z)\lambda_{n}-b_1(z))dz}\leq c_3||F_1||_{C([0,T],L^2(0,1))},
    \end{equation}
    \begin{equation}\label{eq29}
        \sum\limits_{n=1}^{\infty}\dfrac{1}{\sqrt{\lambda_{2n-1}}}F_{1,2n-1}(\tau)e^{-\int_{\tau}^t(a_{1}(z)\lambda_{2n-1}-b_2(z))dz}\leq c_4||F_1||_{C([0,T];L^2(0,1))}
    \end{equation}
    \begin{equation}\label{eq30}
        \sum\limits_{n=1}^{\infty}\lambda_{n}\sqrt{\lambda_{n}}\Phi_{1,n}e^{-\int_0^t(a_1(\tau)\lambda_{n}-b_1(\tau))d\tau}\leq c_5 ||\Phi_1||_{C^4[0,1]},
    \end{equation}
    \begin{equation}\label{eq31}
        \sum\limits_{n=1}^{\infty}\lambda_{n}\sqrt{\lambda_{n}}F_{1,n}(\tau)e^{-\int_{\tau}^t(a_1(z)\lambda_{n}-b_1(z))dz}\leq c_6 ||F_1||_{C([0,T]; L^2(0,1))},
    \end{equation}
    \begin{equation}\label{eq32}
        \sum\limits_{n=1}^{\infty}\sqrt{\lambda_{n}}F_{1,n}(t)\leq c_7||F_1||_{C([0,T];L^2(0,1))},\quad \sum\limits_{n=1}^{\infty}\dfrac{1}{\sqrt{\lambda_{2n-1}}}F_{1,2n-1}(t)\leq c_8 ||F_1||_{C([0,T];L^2(0,1))}
    \end{equation}
for some positive constants $c_i$, $i=1,2,3,..,8$, where $\Phi_{1,n}=\int_0^1\Phi_{1}(\xi_1)\phi_{1,n}(\xi_1)d\xi_1$ and $F_{1,n}(t)=\int_0^1F_{1}(\xi_1,t)\phi_{1,n}(\xi_1)d\xi_1$.
\end{lemma}
\begin{proof}
    The proof studied widely in \cite{20}. Assuming that $0<a_{1m}<a_1(\tau)<a_{1M}$ and $||b_1(\tau)||_{C[0,T]}<b_{1M}$ By integration by parts four times, Cauchy-Schwartz and Bessel inequality we can provide
    \begin{equation*}
        \begin{split}
            &\sum_{n=1}^{\infty}\sqrt{\lambda_n}|\Phi_{1,n}e^{-\int_0^t(a_1(\tau)\lambda_{n}-b_1(\tau))d\tau}|\leq C_2\sum_{n=1}^{\infty}\left(\frac{1}{|\lambda_n|}\right)^{1/2}\left(\sum_{n=1}^{\infty}\bigg|\lambda_n\Phi_{1,n}e^{-a_{1m}\lambda_n}\bigg|\right)^{1/2}\\
            &\leq C_2\sum_{n=1}^{\infty}\left(\frac{1}{|\lambda_n|}\right)^{1/2}\left(\sum_{n=1}^{\infty}\bigg|\lambda_n\Phi_{1,n}\bigg|\right)^{1/2}\leq C||\Phi_{1,n}^{\prime\prime\prime}||_{L^2(0,1)}\leq C||\Phi_{1,n}||_{C^4[0,1]}.
        \end{split}
    \end{equation*}
    Others can be proved analogously.
\end{proof}

\begin{lemma}\label{lem2}
    If $\Phi_2\in C^4[0,1]$ and $F_2(\xi,t)\in C([0,T];L^2(0,1))$ such that $\Phi_{2}(0)=\Phi_{2}(1)=\Phi_{2}'(0)=\Phi_{2}'(1)=\Phi_{2}''(0)=\Phi_{2}''(1)=\Phi_{2}'''(0)=\Phi_{2}'''(1)=0$, $F_2(0,t)=F_2(1,t)=F_2'(0,t)=F_2'(1,t)=F_2''(0,t)=F_2''(1,t)=0$ and $0<s_m<s(t)<s_M$, where $s(t)\in C^2[0,T]$, then following assumptions hold for all $t\in[0,T]$ that are
    \begin{equation}\label{eq33}
        \sum\limits_{n=1}^{\infty}\sqrt{\lambda_{n}}\Phi_2e^{-\int_0^t(a_2(\tau)\lambda_{n}-b_2(\tau))d\tau}\leq c_9||\Phi_2||_{C^4[0,1]},
    \end{equation}
    \begin{equation}\label{eq34}
        \sum\limits_{n=1}^{\infty}\lambda_{n}\sqrt{\lambda_{n}}\Phi_{2,n}e^{-\int_0^t(a_2(\tau)\lambda_{n}-b_2(\tau))d\tau}\leq c_{10}||\Phi_2||_{C^4[0,1]},
    \end{equation}
    \begin{equation}\label{eq35}
        \sum\limits_{n=1}^{\infty}\sqrt{\lambda_{n}}F_{2,n}(\tau)e^{\int_{\tau}^t(a_2(z)\lambda_{n}-b_2(z))dz}\leq c_{11}||F_2||_{C([0,T];L^2(0,1))},
    \end{equation}
    \begin{equation}\label{eq36}
        \sum\limits_{n=1}^{\infty}\lambda_{n}\sqrt{\lambda_{n}}F_{2,n}(\tau)e^{\int_{\tau}^t(a_2(z)\lambda_{n}-b_2(z))dz}\leq c_{12}||F_2||_{C([0,T];L^2(0,1))},
    \end{equation}
    \begin{equation}\label{eq37}
        \sum\limits_{n=1}^{\infty}\sqrt{\lambda_{n}}F_{2,n}(t)\leq c_{13}||F_2||_{C([0,T];L^2(0,1))},
    \end{equation}
for some positive constants $c_i$, $i=9,10,...,13$, where $\Phi_{2,n}=\int_0^1\Phi_{2}(\xi_2)\phi_{2,n}(\xi_2)d\xi_2$ and $F_{2,n}(t)=\int_0^1F_{2}(\xi_2,t)\phi_{2,n}(\xi_2)d\xi_2$.
\end{lemma}
\begin{proof}
    We can verify these inequalities by similar approach in the Lemma \ref{lem1}.
\end{proof}

\begin{lemma}\label{lem3}
    Assume there exists positive constants $a_{im}$, $a_{iM}$, $b_{iM}$ and $M_{i}$ such that $0<a_{im}<a_i(t)<a_{iM}$ and $||b_i(t)||<b_{iM}$ for all $t\in[0,T]$, and $R_i(t)\in C[0,T]$ with $||R_i(t)||_{C[0,T]}\leq M_{i}$, and assumptions in Lemma \ref{lem1} and Lemma \ref{lem2} hold, namely $\Phi_{i,n}\in C^4[0,1]$ and $F_{i,n}(t)\in C([0,T]; L^2(0,1))$ together with the compatibility conditions. Then the series \eqref{eq22} converges absolutely and uniformly on $[0,1]\times[0,T]$, and there exists a constant $C>0$, independent of $(\xi_i,t)$, such that
    \begin{equation}\label{uniform}
        \begin{split}
            &\sup\limits_{0\leq t\leq T}||V_{i}(\cdot,t)||_{L^2(0,1)}+\sup\limits_{0\leq t\leq T}||V_i(\cdot,t)||_{\mathbb{H}_0^1(0,1)}+\sup\limits_{0\leq t\leq T}||V_{i,t}(\cdot,t)||_{L^2(0,1)}\\
            &\leq C\left(||\Phi_{i}||_{C^4[0,1]}+||R_i(t)||_{C[0,T]}||F_i||_{C([0,T];L^2(0,1))}\right).
        \end{split}
    \end{equation}
\end{lemma}

\begin{proof}
    At first, let us prove the boundedness of the time-depended coefficients $R_i(t)$ which defined by \eqref{eq24} and \eqref{eq26}, and we can rewrite them in the integral form as
    $$R_i(t)=G_i(t)+\int_0^t K_i(\tau,t)R_i(\tau)d\tau.$$
    Assuming assumptions in Lemma \ref{lem1} and Lemma \ref{lem2} hold and let consider that $G_i(t)\in C[0,T]$, $K(\tau,t)\in C([0,T]\times[0,T])$ are bounded $|G_i(t)|\leq M_{G_i}$, $|K(\tau,t)|\leq M_{K_i}$ respectively, then we have
    $$|R_i(t)|\leq M_{G_i}+M_{K_i}\int_0^t|R_i(\tau)|d\tau.$$
    Applying Gronwall's inequality gives
    $$|R_i(t)|\leq M_{G_i}e^{M_{K_i}t}\leq M_{G_i}e^{M_{K_i}T},\quad\quad t\in[0,T].$$
    Hence, $||R_i(t)||_{C[0,T]}\leq M_i:=M_{C_i}e^{M_{K_i}T}$ which proves that $R_i(t)$ is uniformly bounded on $[0,T]$.

    Since $|\phi_{i,n}(\xi_i)|\leq \sqrt{2}$, then for the Fourier series \eqref{eq19} we have $|V_{i}(\xi_i,t)|\leq \sqrt{2}\sum\limits_{n=1}^{\infty}|V_{i,n}(t)|$. Furthermore, $||V_i(\cdot,t)||_{L^2(0,1)}^2\leq \sum\limits_{n=1}^{\infty}|V_{i,n}(t)|^2,$ then by Parseval's identity. Applying the estimates of Lemmas \ref{lem1} and \ref{lem2} yields
    $$\sup\limits_{0\leq t\leq T}||V_{i}(\cdot,t)||_{L^2(0,1)}\leq C_1\left(||\Phi_{i}||_{C^4[0,1]}+||R_i(t)||_{C[0,T]}||F_i||_{C([0,T];L^2(0,1))}\right).$$
    Similarly, $||V_i(\cdot,t)||_{\mathbb{H}_0^1(0,1)}^2\leq\sum\limits_{n=1}^{\infty}\lambda_n^2|V_{i,n}(t)|^2$ and Lemmas \ref{lem1} and \ref{lem2} guarantee that 
    $$\sup\limits_{0\leq t\leq T}||V_{i}(\cdot,t)||_{\mathbb{H}_0^1 (0,1)}\leq C_2\left(||\Phi_{i}||_{C^4[0,1]}+||R_i(t)||_{C[0,T]}||F_i||_{C([0,T];L^2(0,1))}\right).$$
    Now, if we take $V_{i,n}^{\prime}(t)=-(a_i(t)\lambda_n-b_1(t))V_{i,n}(t)+R_i(t)F_{i,n}(t)$ and using $(x+y+z)^2\leq 3(x^2+y^2+z^2)$, it implies $$||V_{i,t}(\cdot,t)||_{L^2(0,1)}^2\leq \sum\limits_{n=1}^{\infty}|V_{i,n}^{\prime}(t)|^2\leq 3a_i^2(t)\sum\limits_{n=1}^{\infty}\lambda_n^2|V_{i,n}(t)|^2+3b_1^2(t)\sum\limits_{n=1}^{\infty}|V_{i,n}(t)|^2+3|R_i^2(t)|\sum\limits_{n=1}^{\infty}|F_{i,n}(t)|^2$$ Again, by Lemmas \ref{lem1} and \ref{lem2} together with all previous results, assuming $0<a_{im}<a_i(t)<a_{iM}$ and $|b_i(t)|<b_{iM}$ for all $t\in[0,T]$, we deduce that
    $$\sup\limits_{0\leq t\leq T}||V_{i,t}(\cdot,t)||_{L^2(0,1)}\leq C_3\left(||\Phi_{i}||_{C^4[0,1]}+||R_i(t)||_{C[0,T]}||F_i||_{C([0,T];L^2(0,1))}\right).$$
\end{proof}
Finally, we obtain the following the main result.
\begin{theorem}\label{thm1}
    Suppose we have the following assumptions for the given function
    \begin{itemize}
        \item[(i)] $\Phi_i\in C^4[0,1]$ with $\Phi_i^{(j)}(0)=\Phi_i^{(j)}(1)=0$ where $j=0,1,..,4$ and $\Phi_{i,1}>0$, $\Phi_{i,n}\geq 0$ where $i=1,2$, $n\in\mathbb{N}$;
        \item[(ii)] $F_i\in C(\Omega)$ and $F_{i}(\cdot,t)\in C^3[0,1]$ such that $F_i^{(j)}(0,t)=F_i^{(j)}(1,t)=0$ where $j=0,1,2,3$. $F_{i,1}(t)>0$ and $F_{i,n}(t)\geq 0$ for all $t\in[0,T]$, where $i=1,2$;
        \item[(iii)] $g_i(t)\in C^1[0,T]$ with $0<g_{im}<g_i(t)<g_{iM}$ and $0<g_{im'}<g_i'(t)<g_{iM'}$ where $g_{iM},g_{iM'}>0$ and $i=1,2$;
        \item[(iv)] $s(t)\in C^2[0,T]$ such that $s(0)=s_0$ and $s(t)\neq d\neq 0$, $s(t)>0$ for $\forall t\in [0,T]$ and $0<s_m<s(t)<s_M$, $0<s_{m'}<s'(t)<s_{M'}$ for $s_m,s_{m'},s_M,s_{M'}>0$;
        \item[(v)] $\omega\in C^1[0,T]$ with $\omega(t)\neq 0$ for all $t\in[0,T]$ such that $\omega(0)=\int_0^1V_1(\xi_1,t)d\xi_1+g_1(0)+\dfrac{u^*-g_1(0)}{2}$,  $0<\omega_{m}<\omega(t)<\omega_{M}$ and $0<\omega_{m'}<\omega'(t)<\omega_{M'}$ where $\omega_m,\omega_{M},\omega_{m'}$ and $\omega_{M'}$ are positive real numbers;
        \item[(vi)] $a_i(t),b_i(t)\in C[0,T]$ such that $0<a_{im}<\min\limits_{t\in[0,T]}|a_i(t)|$ and $|b_i(t)|\leq b_{iM}$ where $i=1,2$, 
    \end{itemize}
    and the inequalities of Lemmas \ref{lem1} and \ref{lem2} hold, then problem \eqref{eq10} has a unique solution pair $(V_i,R_i)$ where $i=1,2$.
\end{theorem}
\begin{proof}
    The proof can be shown by similar approach in Lemma 2 from \cite{20} that there exist the weak solutions $V_1(\xi_1,t)$ and $V_2(\xi_2,t)$ to the problem \eqref{eq10} for any $R_i(t)\in C[0,T],\quad i=1,2$. Additionally, we can easily conclude that $R_1(t)$ and $R_2(t)$ are continuously differentiable employing uniform convergency of the series in Lemmas \ref{lem1} and \ref{lem2} and assumptions (i)-(vi).    
\end{proof}
\textbf{Remark 1.} If there exist the unique solutions pair of the problem \eqref{eq10}, thus the problem \eqref{eq1} has the unique solution pairs $(u_1,R_1)$ with $\omega\in C^1[0,T]$ such that $\omega(0)=\int_0^{s(0)}u_1(x,t)dx$ and $(u_2,R_2)$.

\subsection{Single observation data.}\label{sec2} Now we choose a point $q\in (0,s(t))$ such that eigenfunction $\phi_{1,n}(q)$ which is solution of the spectral problem \eqref{eq11} for the first phase is bounded for all $n\in\mathbb{N}$ and $\phi_{1,n}(q)\neq 0$ for some $n_0\in\mathbb{N}$. We take $\mathbb{N}_q$ as maximal subset of $\mathbb{N}$ that is $\phi_{1,n}(q)\neq 0$ for $\forall n\in\mathbb{N}_q$ and we use this point as observation point for 
\begin{equation}\label{eq38}
    u_1(q,t)=p(t),
\end{equation}
where $p(t)\in C^1[0,T]$ is a given function, such that $p(t)\neq 0$ for $\forall t\in[0,T]$, where we need to restoring the time-dependent coefficient $R_1(t)$ of source function to the first phase of the problem \eqref{eq1}.

Applying transformation \eqref{eq2} we can rewrite condition \eqref{eq38} in the form
\begin{equation}\label{eq39}
    v_1(q,t)=p(t)-g_1(t)-\dfrac{u^*-g_1(t)}{s(t)}q 
\end{equation}
and function \eqref{eq8} gives us 
\begin{equation}\label{eq40}
    V_1\left(\dfrac{q}{s(t)},t\right)=p(t)-g_1(t)-\dfrac{u^*-g_1(t)}{s(t)}q.
\end{equation}
Using definition of the function $V_1(\xi_1,t)$ defined by \eqref{eq22} and condition \eqref{eq40} we obtain
\begin{equation}\label{eq41}
    \begin{split}
        &\sum\limits_{n=1}^{\infty}\left(\Phi_{1,n}e^{-\int_0^t(a_1(\tau)\lambda_{1,n}-b_1(\tau))d\tau}+\int_0^tR_1(\tau)F_{1,n}(\tau)e^{-\int_{\tau}^t(a_1(z)\lambda_{1,n}-b_1(z))dz}d\tau\right)\phi_{1,n}\left(\dfrac{q}{s(t)}\right)\\&=p(t)-g_1(t)-\dfrac{u^*-g_1(t)}{s(t)}q.
    \end{split}
\end{equation}
Differentiating both sides of \eqref{eq41} respect to $t$, we get
\begin{equation}\label{eq42}
    \begin{split}
        &\sum\limits_{n=1}^{\infty}\Bigg((b_1(t)-a_1(t)\lambda_{1,n})\Phi_{1,n}e^{-\int_0^t(a_1(\tau)\lambda_{1,n}-b_1(\tau))d\tau}\\
        &+(b_1(t)-a_1(t)\lambda_{1,n})\int_0^t R_1(\tau)F_{1,n}(\tau)e^{-\int_{\tau}^t(a_1(z)\lambda_{1,n}-b_1(z))dz}d\tau\\&+R_1(t)F_{1,n}(t)\Bigg)\phi_{1,n}\left(\dfrac{q}{s(t)}\right)-\dfrac{qs'(t)}{s^2(t)}\sum\limits_{n=1}^{\infty}\Bigg(\sqrt{\lambda_{1,n}}\Phi_{1,n}e^{-\int_0^t(a_1(\tau)\lambda_{1,n}-b_1(\tau))d\tau}\\&+\int_0^t\sqrt{\lambda_{1,n}}R_1(\tau)F_{1,n}(\tau)e^{-\int_{\tau}^t(a_1(z)\lambda_{1,n}-b_1(z))dz}\Bigg)\phi_{1,n}'\left(\dfrac{q}{s(t)}\right)\\&=p'(t)-g_1'(t)+\dfrac{g_1'(t)s(t)+2s'(t)(u^*-g_1(t))}{s^3(t)}q,
    \end{split}
\end{equation}
it implies that solution function for $R_1(t)$ takes the form
\begin{equation}\label{eq43}
    \begin{split}
        R_1(t)&=\dfrac{1}{\sum\limits_{n=1}^{\infty}F_{1,n}(t)\phi_{1,n}\left(\frac{q}{s(t)}\right)}\Bigg[a_1(t)\sum\limits_{n=1}^{\infty}\lambda_{1,n}\Phi_{1,n}\phi_{1,n}\left(\frac{q}{s(t)}\right)e^{-\int_0^t(a_1(\tau)\lambda_{1,n}-b_1(\tau))d\tau}\\
        &-b_1(t)\sum\limits_{n=1}^{\infty}\Phi_{1,n}\phi_{1,n}\left(\frac{q}{s(t)}\right)e^{-\int_0^t(a_1(\tau)\lambda_{1,n}-b_1(\tau))d\tau}\\
        &+a_1(t)\int_0^t\left(\sum\limits_{n=1}^{\infty}\lambda_{1,n}F_{1,n}(\tau)\phi_{1,n}\left(\frac{q}{s(t)}\right)e^{-\int_{\tau}^t(a_1(z)\lambda_{1,n}-b_1(z))dz}\right)R_1(\tau)d\tau\\
        &-b_1(t)\int_0^t\left(\sum\limits_{n=1}^{\infty}F_{1,n}(\tau)\phi_{1,n}\left(\frac{q}{s(t)}\right)e^{-\int_{\tau}^t(a_1(z)\lambda_{1,n}-b_1(z))dz}\right)R_1(\tau)d\tau\\
        &+\dfrac{qs'(t)}{s^2(t)}\sum\limits_{n=1}^{\infty}\sqrt{\lambda_{1,n}}\Phi_{1,n}\phi_{1,n}^{\prime}\left(\frac{q}{s(t)}\right)e^{-\int_0^t(a_1(\tau)\lambda_{1,n}-b_1(\tau))d\tau}\\
        &+\dfrac{qs'(t)}{s^2(t)}\int_0^t\left(\sum\limits_{n=1}^{\infty}\sqrt{\lambda_{1,n}}F_{1,n}(\tau)\phi_{1,n}^{\prime}\left(\frac{q}{s(t)}\right)e^{-\int_{\tau}^t(a_1(z)\lambda_{1,n}-b_1(z))dz}\right)R_1(\tau)d\tau\\
        &+p'(t)-g_1'(t)+\dfrac{g_1'(t)s(t)-s'(t)(u^*-g_1(t))}{s^2(t)}\Bigg]
    \end{split}
\end{equation}
To provide the existence and uniqueness of the solution of the problem \eqref{eq10} with condition of single data, we need to add the following assumptions:
\begin{itemize}
    \item[(vii)] $\Phi_1\in C^4[0,1]$ with $\Phi_{1,n}\phi_{1,n}\left(\frac{q}{s(t)}\right)\geq 0$ for any $n\in\mathbb{N}_q$, $\forall t\in[0,T]$ and $\Phi_{1,n_0}\phi_{1,n_0}\left(\frac{q}{s(t)}\right)>0$ for some $n_0\in\mathbb{N}_q$ and $\forall t\in[0,T]$;
    \item[(viii)] $F_{1}(\xi_1,t)\in C([0,1]\times[0,T])$ with $F_{1,n}(t)\phi_{1,n}\left(\frac{q}{s(t)}\right)\geq 0$ for $\forall t\in[0,T]$ and for any $n\in\mathbb{N}_q$;
    \item[(ix)] $p(t)\in C^1[0,T]$ with $p(t)\neq 0$ and $p(0)=\Phi_{1}(q)+g_1(0)+\frac{u^*-g_1(0)}{s^2(0)}q$, and $0<p_m<p(t)<p_M$ where $p_m,p_M>0$,
\end{itemize}
then we can verify the following theorem.
\begin{theorem}\label{thm2}
    Assuming conditions in Theorem \ref{thm1} and conditions (vi)-(ix) hold, then there exists a unique weak solution pairs $(V_1,R_1)$  with $p(t)\in C^1[0,T]$ and $(V_2,R_2)$ which defined by \eqref{eq22}, \eqref{eq25} and \eqref{eq43} to the the problem \eqref{eq10}. 
\end{theorem}
The proof can be established similarly as in Lemma 2 from \cite{20} and the function $R_1(t)$ is continuously differentiable if the series $\sum_{n=1}^{\infty}\Phi_{1,n}e^{-\int_0^t(a_1(\tau)\lambda_{1,n}-b_1(\tau))d\tau}\phi_{1,n}(q)$ is uniformly convergent as $|a_1(t)|>m_1>0$, $|b_1(t)|\leq M_2$  for any $t\in [t_0,T]$ and $\left|\phi_{1,n}\left(\frac{q}{s(t)}\right)\right|\leq M_3$ for some $M_i=const>0$ and $\forall t\in[0,T]$ where $i=1,2,3$. Then it implies that majorant series $M\sum_{n=1}^{\infty}|\Phi_{1,n}|$ and $M\sum_{n=1}^{\infty}\lambda_{1,n}|\Phi_{1,n}|$ are convergent.

\subsection{Non-local data.}\label{sec3} Now let us consider inverse problem of finding solution pair $(u_1,R_1)$ for the first phase such that $R_1(t)\in C[0,T]$ with non-local data 
\begin{equation}\label{eq44}
    \int_0^{s(t)}\omega(x)u_1(x,t)dx=h(t),\quad t\in[0,T],
\end{equation}
where $\omega(x)$ and $h(t)$ are given functions such that $\omega(x)\in L^2(0,s(t))$ and $h(t)\in C^1[0,T]$. 

Applying transformation \eqref{eq2}, the condition \eqref{eq44} takes the form
\begin{equation}\label{eq45}
    \int_0^{s(t)}\omega(x)v_1(x,t)dx+g_1(t)\int_0^{s(t)}\omega(x)dx+\dfrac{u^*-g_1(t)}{s(t)}\int_0^{s(t)}x\omega(x)dx=h(t).
\end{equation}
Hence, using the change of variable with function \eqref{eq8}, it becomes 
\begin{equation}\label{eq46}
    \int_0^1 \omega(\xi_1)V_1(\xi_1,t)d\xi_1=\dfrac{h(t)}{s(t)}-g_1(t)\int_0^1\omega(\xi_1)d\xi_1-\dfrac{u^*-g_1(t)}{s(t)}\int_0^1\xi_1\omega(\xi_1)d\xi_1.
\end{equation}
Substituting the solution function $V_1(\xi_1,t)$ of the problem \eqref{eq10} defined by \eqref{eq22} into \eqref{eq46} we obtain
\begin{equation}\label{eq47}
    \begin{split}
        &\sum\limits_{n=1}^{\infty}\Bigg(\Phi_{1,n}e^{-\int_0^t(a_1(\tau)\lambda_{1,n}-b_1(\tau))d\tau}+\int_0^tR_1(\tau)F_{1,n}(\tau)e^{-\int_{\tau}^t(a_1(z)\lambda_{1,n}-b_1(z))dz}\Bigg)\int_0^1\omega(\xi_1)\phi_{1,n}(\xi_1)d\xi_1\\
        &=\dfrac{h(t)}{s(t)}-g_1(t)\int_0^1\omega(\xi_1)d\xi_1-\dfrac{u^*-g_1(t)}{s(t)}\int_0^1\xi_1\omega(\xi_1)d\xi_1
    \end{split}
\end{equation}
To this end, differentiating both sides of \eqref{eq47} respect to $t$, we have
\begin{equation}\label{eq48}
    \begin{split}
        &\sum\limits_{n=1}^{\infty}\Bigg((b_1(t)-a_1(t)\lambda_{1,n})\Phi_{1,n}e^{-\int_0^t(a_1(\tau)\lambda_{1,n}-b_1(\tau))d\tau}\\
        &+(b_1(t)-a_1(t)\lambda_{1,n})\int_0^tR_1(\tau)F_{1,n}(\tau)e^{-\int_{\tau}^t(a_1(z)\lambda_{1,n}-b_1(z))dz}+R_1(t)F_{1,n}(t)\Bigg)\int_0^1\omega(\xi_1)\phi_{1,n}(\xi_1)d\xi_1\\
        &=\dfrac{h'(t)s(t)-h(t)s'(t)}{s^2(t)}-g_1'(t)\int_0^1\omega(\xi_1)d\xi_1+\dfrac{g_1'(t)s(t)+s'(t)(u^*-g_1(t))}{s^2(t)}\int_0^1\xi_1\omega(\xi_1)d\xi_1,
    \end{split}
\end{equation}
taking into account that $\left(\omega(\xi_1),\phi_{1,n} (\xi_1)\right)_{L^2(0,1)}=\int_0^1\omega(\xi_1)\phi_{1,n}(\xi_1)d\xi_1=\omega_{1,n}$, one can easily determine the solution function $R_1(t)$
\begin{equation}\label{eq49}
    \begin{split}
        R_1(t)&=\dfrac{1}{\sum_{n=1}^{\infty}F_{1,n}\omega_{1,n}}\Bigg[a_1(t)\sum\limits_{n=1}^{\infty}\lambda_{1,n}\Phi_{1,n}\omega_{1,n}e^{-\int_0^t(a_1(\tau)\lambda_{1,n}-b_1(\tau))d\tau}\\
        &-b_1(t)\sum\limits_{n=1}^{\infty}\Phi_{1,n}\omega_{1,n}e^{-\int_0^t(a_1(\tau)\lambda_{1,n}-b_1(\tau))d\tau}\\
        &+a_1(t)\int_0^t\left(\sum\limits_{n=1}^{\infty}\lambda_{1,n}F_{1,n}(\tau)\omega_{1,n}e^{-\int_{\tau}^t(a_1(z)\lambda_{1,n}-b_1(z))dz}\right)R_1(\tau)d\tau\\
        &-b_1(t)\int_0^t\left(\sum\limits_{n=1}^{\infty}F_{1,n}(\tau)\omega_{1,n}e^{-\int_{\tau}^t(a_1(z)\lambda_{1,n}-b_1(z))dz}\right)R_1(\tau)d\tau\\
        &+\dfrac{h'(t)s(t)-h(t)s'(t)}{s^2(t)}-g_1'(t)\int_0^1\omega(\xi_1)d\xi_1+\dfrac{g_1'(t)s(t)+s'(t)(u^*-g_1(t))}{s^2(t)}\int_0^1\xi_1\omega(\xi_1)d\xi_1\Bigg].
    \end{split}
\end{equation}
We have $\omega(\xi_1)\in L^2(0,1)$ and assuming that $C_1:=\int_0^1\omega(\xi_1)d\xi_1$ and $C_2:=\int_0^1\xi_1\omega(\xi_1)d\xi_1$ in \eqref{eq49} and $\sum_{n=1}^{\infty}|\omega_{1,n}|\leq C_3$, then by Cauchy-Schwartz and Bessel inequalities and inequalities in Lemma \ref{lem1} and Lemma \ref{lem2} we have
$$\sum_{n=1}^{\infty}\lambda_{n}\Phi_{1,n}\omega_{1,n}e^{-\int_0^t(a_1(\tau)\lambda_{1,n}-b_1(\tau))d\tau}\leq K_1\left(\sum_{n=1}^{\infty}|\lambda_{n}\Phi_{1,n}|^2\right)^{1/2}\left(\sum_{n=1}^{\infty}|\omega_{1,n}|^2\right)^{1/2}\leq M_1||\Phi_{1,n}||_{C^4[0,1]},$$
$$\sum_{n=1}^{\infty}\lambda_{n}F_{1,n}(\tau)\omega_{1,n}e^{-\int_0^t(a_1(\tau)\lambda_{1,n}-b_1(\tau))d\tau}\leq K_1\left(\sum_{n=1}^{\infty}|\lambda_{n}F_{1,n}(\tau)|^2\right)^{1/2}\left(\sum_{n=1}^{\infty}|\omega_{1,n}|^2\right)^{1/2}$$
$$\leq M_2||F_{1,n}||_{C([0,T];L^2(0,1)},$$
are uniformly convergent, where $K_1=C_3e^{b_M T}$. Hence, integral solution \eqref{eq49} has a unique continuous solution.

Finally, we have the following main result.
\begin{theorem}\label{thm3}
    Let assumptions in Theorem \ref{thm1}, the statements (vi)-(ix) and the following conditions
    \begin{itemize}
        \item[1.] $\omega(\xi_1)\in L^2(0,1)$ and $\int_0^1\omega(\xi_1)\phi_{1,n}(\xi_1)d\xi_1=\omega_{1,n}$, and $\sum_{n=1}^{\infty}|\omega_{1,n}|\leq C_3$;
        \item[2.] $h(t)\in C^1[0,T]$ with $h(t)\neq 0$ for $\forall t\in[0,T]$ and $h(0)=\int_0^1\omega(\xi_1)\Phi_1(\xi_1)d\xi_1+g_1(0)\int_0^1\omega(\xi_1)d\xi_1-\dfrac{u^*-g_1(0)}{s(0)}\int_0^1\xi_1\omega(\xi_1)d\xi_1$ where $g_1(0)\neq 0$ and $s(0)=s_0$. 
    \end{itemize}
    Then, there exists a unique solution pair of function $(V_1,R_1)$ to the first phase of the problem \eqref{eq10}.
\end{theorem}
\begin{proof}
    The proof can be shown by analogously in Lemma 2 from \cite{20} that there exist the weak solution $V_1(\xi_1,t)$ to the first phase of the problem \eqref{eq10} for any $R_1(t)\in C[0,T]$. It implies that $R_1(t)$ is continuously differentiable employing uniform convergency of the series in Lemmas \ref{lem1} and \ref{lem2}, assumptions in Theorem \ref{thm1}, (vi)-(viii) mentioned in the section \ref{sec1} and conditions in Theorem \ref{thm3}.
\end{proof}
\textbf{Remark 2.} If there exist the unique solutions pair $(V_1,R_1)$ of the problem \eqref{eq10}, thus the problem \eqref{eq1} has the unique solution pairs $(u_1,R_1)$ with $h(x)\in L^2(0,s(t))$ such that $h(0)=\int_0^{s(0)}h(x)u_1(x,t)dx$.

\textbf{Example 1.} Let us consider the following inverse problem in the domains $\Omega_1=(0,\pi(t+1))$ and $\Omega_2=(\pi(t+1),d)$ with suitable choice of heat sources $f_1(x,t),\;f_2(x,t)$, with given boundary functions $g_1(t),\;g_2(t)$ and with initial functions $\varphi_1(x),\;\varphi_2(x)$ for determining pairs $(u_1(x,t),R_1(t))$ and $(u_2(x,t),R_2(t))$:
\begin{equation}\label{ex-1}
    \begin{cases}
        \partial_t u_1-a_1^2\partial_{xx}u_1=R_1(t)f_1(x,t),&(x,t)\in\Omega_1,\\
        \partial_t u_2-a_2^2\partial_{xx}u_2=R_2(t)f_2(x,t),&(x,t)\in\Omega_2,\\
        u_1(x,0)=\sqrt{2}\sin(x)+1+\dfrac{u^*-1}{\pi}x,\quad s(0)=\pi,&\forall x\in[0,\pi],\\
        u_2(x,0)=\sqrt{2}\sin\left(\dfrac{x-\pi}{\frac{d}{\pi}-1}\right)+u^*+\dfrac{1-u^*}{d-\pi}(x-\pi),\quad s(0)=\pi,&\forall x\in[\pi,d],\\
        u_1(0,t)=t+1,&\forall t\in[0,T],\\
        u_1(\pi(t+1),t)=u^*,&\forall t\in[0,T],\\
        u_2(\pi(t+1),t)=u^*,&\forall t\in[0,T],\\
        -k_1\partial_x u_1(\pi(t+1),t)=-k_2\partial_x u_2(\pi(t+1),t)+L\pi,&\forall t\in[0,T],\\
        u_2(d,t)=t^2+1,&\forall t\in[0,T],
    \end{cases}
\end{equation}
where 
$$f_1(x,t)=\sqrt{2(t+1)}e^{\frac{a_1^2t}{t+1}}\sin\left(\frac{x}{t+1}\right)+\dfrac{\pi(t+1)^2-u^*x}{\pi(t+1)^2R_1(t)},$$

\begin{equation}
    \begin{split}
        f_2(x,t)&=\sqrt{\frac{2(d-\pi)}{d-\pi(t+1)}}e^{\frac{a_2^2\pi^2t}{(d-\pi)(d-\pi(t+1))}}\sin\left(\frac{x-\pi(t+1)}{\frac{d}{\pi}-t-1}\right)\\
        &+\frac{2t(d-\pi)+\pi(1-u^*)-\pi t^2}{R_2(t)(d-\pi(t+1))^2}(x-\pi(t+1))-\frac{\pi(t^2+1-u^*)}{R_2(t)(d-\pi(t+1))}.
    \end{split}
\end{equation}
with integral overdetermination condition
\begin{equation}\label{ex-2}
    \int_0^{\pi(t+1)}u_1(x,t)dx=\pi(t+1)^2.
\end{equation}
Introducing transformations
\begin{equation}\label{ex-3}
    v_1(x,t)=u_1(x,t)-(t+1)-\frac{u^*-(t+1)}{\pi(t+1)}x,
\end{equation}
\begin{equation}\label{ex-4}
    v_2(x,t)=u_2(x,t)-u^*-\frac{t^2+1-u^*}{d-\pi(t+1)}(x-\pi(t+1)),
\end{equation}
the problem \eqref{ex-1} changes to the following form:
\begin{equation}\label{ex-5}
    \begin{cases}
        \partial_tv_1-a_1^2\partial_{xx}v_1=R_1(t)\sqrt{2(t+1)}e^{\frac{a_1^2t}{t+1}}\sin\left(\frac{x}{t+1}\right),&(x,t)\in\Omega_1,\\
        \partial_tv_2-a_2^2\partial_{xx}v_2=R_2(t)\sqrt{\frac{2(d-\pi)}{d-\pi(t+1)}}e^{\frac{a_2^2\pi^2t}{(d-\pi)(d-\pi(t+1))}}\sin\left(\frac{x-\pi(t+1)}{\frac{d}{\pi}-t-1}\right),&(x,t)\in\Omega_1,\\
        v_1(x,0)=\sqrt{2}\sin(x),\quad s(0)=\pi,&\forall x\in[0,\pi],\\
        v_2(x,0)=\sqrt{2}\sin\left(\frac{x-\pi}{\frac{d}{\pi}-1}\right),\quad s(0)=\pi,&\forall x\in[\pi,d],\\
        v_1(0,t)=v_1(\pi(t+1),t)=0,&\forall t\in[0,T],\\
        v_2(\pi(t+1),t)=v_2(d,t)=0,&\forall t\in[0,T],\\
        -k_1\partial_xv_1(\pi(t+1),t)=-k_2\partial_x v_2(\pi(t+1),t)+\chi(t)+L\pi,&\forall t\in[0,T].
    \end{cases}
\end{equation}
where $\chi(x)=k_1\frac{u^*-t-1}{\pi(t+1)}-k_2\frac{t^2+1-u^*}{d-\pi(t+1)}$ with integral overdetermination condition
$$\int_0^{\pi(t+1)}v_1(x,t)dx=\frac{t+1-u^*}{2}\pi(t+1).$$

Considering the fixed domain $\Omega=(0,1)\times (0,T)$ and applying the change of variables 
$$v_1(x,t)=V_1(\xi_1,t),\quad \xi_1=\frac{x}{\pi(t+1)}$$
$$v_2(x,t)=V_2(\xi_2,t),\quad \xi_2=\frac{x-\pi(t+1)}{d-\pi(t+1)},$$
we can easily rewrite the problem \eqref{ex-5} in the following form:
\begin{equation}\label{ex-6}
    \begin{cases}
        \dfrac{\partial V_1}{\partial t}-b_1(\xi_1,t)\dfrac{\partial V_1}{\partial \xi_1}-a_1(t)\dfrac{\partial^2 V_1}{\partial \xi_1}=R_1(t)\sqrt{2(t+1)}e^{\frac{a_1^2t}{t+1}}\sin(\pi\xi_1),&(\xi_1,t)\in\Omega,\\
        \dfrac{\partial V_2}{\partial t}-b_2(\xi_2,t)\dfrac{\partial V_2}{\partial \xi_2}-a_2(t)\dfrac{\partial^2 V_2}{\partial \xi_2}=R_2(t)\sqrt{\frac{2(d-\pi(t+1))}{d-\pi}}e^{\frac{a_2^2\pi^2 t}{(d-\pi)(d-\pi(t+1))}}\sin(\pi\xi_2),&(\xi_2,t)\in\Omega,\\
        V_1(\xi_1,0)=\sqrt{2}\sin(\pi\xi_1),&\xi_1\in[0,1],\\
        V_2(\xi_2,0)=\sqrt{2}\sin(\pi\xi_2),&\xi_2\in[0,1],\\
        V_1(0,t)=V_1(1,t)=0,&t\in[0,T],\\
        V_2(0,t)=V_2(1,t)=0,&t\in[0,t],\\
        -\dfrac{k_1}{\pi(t+1)}\dfrac{\partial V_1}{\partial \xi_1}\bigg|_{\xi_1=1}=-\dfrac{k_2}{d-\pi(t+1)}\dfrac{\partial V_2}{\partial \xi_2}\bigg|_{\xi_2=0}+\chi(t)+L\pi,&t\in[0,T].
    \end{cases}
\end{equation}
where $a_1(t)=\frac{a_1^2}{\pi^2(t+1)^2}$, $a_2(t)=\frac{a_2^2}{(d-\pi(t+1))^2}$, $b_1(t)=-\frac{1}{2(t+1)}$, $b_2(t)=-\frac{1}{2(d-\pi(t+1))}$ and $\chi(t)=k_1\frac{u^*-(t+1)}{\pi(t+1)}-k_2\frac{t^2+1-u^*}{d-\pi(t+1)}$ with overdetermination condition
$$\int_0^1V_1(\xi_1,t)d\xi_1=\dfrac{t+1-u^*}{2}.$$

Suppose that eigenvalue and orthonormal eigenfunction to the problem \eqref{eq11} defined by \eqref{eq12}, then one can easily estimate that $\Phi_{1,1}=1$, $F_{1,1}(t)=\sqrt{t+1}e^{\frac{a_1^2t}{t+1}}$, $\Phi_{2,1}=1$, $F_{2,1}(t)=\sqrt{\frac{d-\pi(t+1)}{d-\pi}}e^{\frac{a_2^2\pi^2t}{(d-\pi)(d-\pi(t+1))}}$ and $\Phi_{1,n}=\Phi_{2,n}=F_{1,n}(t)=F_{2,n}(t)=0$ for all $n>1$.

Finally, the solution pairs $(V_1,R_1)$ and $(V_2,R_2)$ to the problem \eqref{ex-6} can be defined as following form:
\begin{equation}\label{ex-7}
    V_1(\xi_1,t)=\dfrac{e^{-\frac{a_1^2t}{t+1}}}{\sqrt{t+1}}\left(1+\int_0^tR_1(\tau)d\tau\right)\sin(\pi\xi_1),
\end{equation}
where
\begin{equation}\label{ex-8}
    R_1(t)=e^{\frac{a_1^2t}{t+1}}\left(\frac{\pi\sqrt{t+1}}{4}+\frac{\pi(t+1-u^*)}{8\sqrt{t+1}}+\frac{a_1^2}{(t+1)^2}\right)
\end{equation}
and 
\begin{equation}\label{ex-9}
    V_2(\xi_2,t)=\sqrt{\frac{d-\pi(t+1)}{d-\pi}}e^{-\frac{a_2^2\pi^2t}{(d-\pi)(d-\pi(t+1))}}\left(1+\int_0^tR_2(\tau)d\tau\right)\sin\left(\pi\xi_2\right),
\end{equation}
where
\begin{equation}\label{ex-10}
    \begin{split}
        R_2(t)&=\Bigg(\left(\frac{a_2^2\pi^2}{(d-\pi(t+1))^2}-\frac{a_1^2}{(t+1)^2}\right)\dfrac{\sqrt{d-\pi(t+1)}}{(t+1)^{3/2}}-\frac{k_1\sqrt{(d-\pi)(d-\pi(t+1))}}{k_2\pi(t+1)^{3/2}}R_1\\
        &-\frac{k_1\sqrt{d-\pi}(2\pi(t+1)-3d)}{2k_2\pi(t+1)^{5/2}}\Bigg)e^{\frac{a_2^2\pi^2t}{(d-\pi)(d-\pi(t+1))}-\frac{a_1^2t}{t+1}}\left(1+\int_0^tR_1(\tau)d\tau\right)+\chi'(t).
    \end{split}
\end{equation}

\section{Global existence of the solution}
Now, we attempt show the global existence and uniqueness from the solution of the cylindrical Problem \eqref{eq10} which implies in the existence and uniqueness of the solution of the noncylindrical Problem \eqref{eq4}, since they are equivalent.

\begin{theorem}\label{thm4}
Let suppose $s_m$, $s_{m^{\prime}}$, $s_M$ and $s_{M^{\prime}}$ are positive real numbers such that $s(t)\in C^2(0,T)$, $s^{\prime}(t)\in L^{\infty}(0,1)$ with $0<s_m<s(t)<s_M$, $s_{m^{\prime}}<s^{\prime}(t)<s_{M^{\prime}}$ and $s^{\prime}(t)\in L^{\infty}(0,T)$ for all $t\in[0,T]$. Moreover, initial datas $\Phi_{i0}\in \mathbb{H}_0^1(0,1)$ and $F_i(\xi_i,t)\in L^1(0,T,L^2(\Omega))\cap L^2(0,T,L^2(0,1))$, for the uniquely reconstructed coefficients $R_i(t)$ obtained in Theorem 1, 2, and Theorem 3, problem \eqref{eq10} possesses a unique strong solution pair $(V_1,V_2)$ of the problem \eqref{eq10}, that is,
$$\frac{\partial V_i}{\partial t}-a_i(t)\frac{\partial^2 V_i}{\partial \xi^2}-b(\xi_i,t)\frac{\partial V_i}{\partial \xi}=R_i(t)F_i(\xi_i,t),\quad\text{where}\quad i=1,2\quad\text{in}\quad L^2(0,T;L^2(0,1)),$$ which satisfy the following statements for both $i=1,2$:
\begin{itemize}
    \item[a)] $V_i\in L^2(0,T;\mathbb{H}_0^1(0,1)\cap \mathbb{H}^2(0,1))\cap L^{\infty}(0,T;\mathbb{H}_0^1(0,1))$,
    \item[b)] $\partial_t V_i\in L^2(0,T;L^2(0,1))$,
    \item[c)] $R_i(t)\in C[0,T]$, $R_i(t)>0$ and $0<m_i<R_i(t)< M_i$ with $m_i,M_i>0$ for all $t\geq 0$. 
\end{itemize}
\end{theorem}
\begin{proof}
    \textbf{Existence.} Let $T>0$ and define the subspace $X_m$ spanned by $\{\phi_1,\phi_2,...,\phi_m\}$, where $\phi_j$ with $j=1,2,3,...m$ are assumed to be eigenvectors of the space $\mathbb{H}_0^1(0,1)$ that is the solution of the spectral problem \eqref{eq11}. Thus, we can deduce that if $V_{i,m}\in X_m$ where $i=1,2$, then it can be written as linear combination of the eigenvectors $\phi_m$ belonging to $X_m$ such that
    $$V_{i,m}(\xi_i,t)=\sum\limits_{j=1}^m V_{i,j,m}(t)\phi_j(\xi_i),\quad\quad i=1,2,$$
    where $V_{i,j,m}$ is the solution of the system of ordinary differential equation with boundary condition as follows
    \begin{equation}\label{eq62}
        \begin{cases}
            \left(V_{i,m}^{\prime},\phi\right)-a_i(t)\left(\frac{\partial^2 V_{i,m}}{\partial \xi_i^2},\phi\right)-\left(b(\xi_i,t)\frac{\partial V_{i,m}}{\partial \xi_i},\phi\right)=R_i(t)\left(F_i(\xi_i,t),\phi\right),& \forall\phi\in X_m,\\
            V_{i,m}(0)=\Phi_{i,m}\rightarrow \Phi_i,& \text{in}\quad \mathbb{H}_0^1(\Omega) 
        \end{cases}
    \end{equation}
    where $V_{i,m}^{\prime}=\frac{\partial V_{i,m}}{\partial t}$. The system \eqref{eq62}  has local solution in the interval $(0,T_m)$, by Caratheodory Theorem. To extend the local solution to the interval $(0,T)$ independent of $m$ the
    following a priori estimate is needed:

    \begin{itemize}
        \item [Case 1.] If $\phi=V_{i,m}$ where $i=1,2$, then equation in the system \eqref{eq62} takes the form
        $$\left(V_{i,m}^{\prime},V_{i,m}\right)-a_i(t)\left(\frac{\partial^2 V_{i,m}}{\partial \xi_i^2},V_{i,m}\right)-\left(b(\xi_i,t)\frac{\partial V_{i,m}}{\partial \xi_i},V_{i,m}\right)=R_i(t)\left(F_i(\xi_i,t),V_{i,m}\right).$$
        The first term gives us
        $$\left(V_{i,m}^{\prime},V_{i,m}\right)=\int_0^1\frac{\partial V_{i,m}}{\partial t}V_{i,m}d\xi_i=\int_0^1\frac{d}{dt}\left(\frac{1}{2}V_{i,m}\right)^2d\xi_i=\frac{1}{2}\frac{d}{dt}\int_0^1 V_{i,m}^2d\xi_i=\frac{1}{2}\frac{d}{dt}\left|V_{i,m}\right|_{L^2(0,1)}^2.$$
        By integration by parts and boundary conditions in \eqref{eq10}, we estimate the following
        $$a_i(t)\left(\frac{\partial^2 V_{i,m}}{\partial \xi_i^2},V_{i,m}\right)=-a_i(t)\int_0^1 \left(\frac{\partial V_{i,m}}{\partial \xi_i}\right)^2 d\xi_i=-a_i(t)||V_{i,m}||^2_{\mathbb{H}_0^1(0,1)}$$
        and analogously for the third term we get
        $$\left(b(\xi_i,t)\frac{\partial V_{i,m}}{\partial \xi_i},V_{i,m}\right)=\int_0^1 b(\xi_i,t)\frac{\partial V_{i,m}}{\partial \xi_i}V_{i,m}d\xi_i=-\frac{s^{\prime}(t)}{2s(t)}\int_0^1 V_{i,m}^2 d\xi_i=-\frac{s^{\prime}(t)}{2s(t)}\left|V_{i,m}\right|^2_{L^2(0,1)}.$$
        Applying these estimations and Holder inequality, we can provide the following inequality
        $$\frac{1}{2}\frac{d}{dt}|V_{i,m}|_{L^2(0,1)}^2+a_i(t)||V_{i,m}||_{\mathbb{H}_0^1(0,1)}^2+\frac{s^{\prime}(t)}{2s(t)}\left|V_{i,m}\right|^2_{L^2(0,1)}\leq \frac{1}{2}R_i(t)\left(|F_i|_{L^2(0,1)}^2+|V_{i,m}|_{L^2(0,1)}^2\right).$$
        Integrating over $[0,t)$ with $t\in[0,T_m)$ and assuming $0<s_m<s(t)<s_M$, $s_{m^{\prime}}<s^{\prime}(t)<s_{M^{\prime}}$, $|a_i(t)|<N_i$  and (c) hold, then again by integration by parts we have
        $$|V_{i,m}|_{L^2(0,1)}^2+\int_0^t ||V_{i,m}||_{\mathbb{H}_0^1(0,1)}^2d\tau\leq\frac{1}{2c_1}|V_{i,m}(0)|_{L^2(0,1)}^2+\frac{c_2}{c_1}|F_{i}|_{L^2(0,1)}2+\frac{c_3}{c_1}\int_0^T |V_{i,m}|_{L^2(0,1)}^2d\tau,$$
        where $c_1=\min\{\frac{1}{2}, N_i\}$, $c_2=\frac{M_i}{2}$ and $c_3=\frac{M_i}{2}+\frac{s_{M^{\prime}}}{2s_m}$. Then according to the Gronwall's inequality we can deduce that
        $$V_{i,m}\;\text{ is bounded in }\;L^{\infty}(0,T; L^2(0,1))\quad\text{and}\quad V_{i,m}\;\text{ is bounded in }\; L^2(0,T;\mathbb{H}_0^1(0,1)).$$

        \item[Case 2.] If we take $\phi=V_{i,m}^{\prime}=\frac{d}{dt}V_{i,m}$ where $i=1,2$, then the system \eqref{eq62} becomes
        $$\left(V_{i,m}^{\prime},V_{i,m}^{\prime}\right)-a_i(t)\left(\frac{\partial^2 V_{i,m}}{\partial \xi_i^2},V_{i,m}^{\prime}\right)-\left(b(\xi_i,t)\frac{\partial V_{i,m}}{\partial \xi_i},V_{i,m}^{\prime}\right)=R_i(t)\left(F_i(\xi_i,t),V_{i,m}^{\prime}\right),$$
        for all $V_{i,m}^{\prime}\in X_m$.
        Therefore, for the all terms on the left-hand side of this identity we have the following verifications
        $$\left(V_{i,m}^{\prime},V_{i,m}^{\prime}\right)=\int_0^1\left(\frac{dV_{i,m}}{dt}\right)^2 d\xi_i=|V_{i,m}^{\prime}|_{L^2(0,1)}^2,$$
        $$a_i(t)\left(\frac{\partial^2 V_{i,m}}{\partial \xi_i^2},V_{i,m}^{\prime}\right)=a_i(t)\int_0^1\frac{\partial^2V_{i,m}}{\partial\xi_i^2}\frac{dV_{i,m}}{dt}d\xi_i=a_i(t)\frac{d}{dt}\int_0^1\frac{\partial^2V_{i,m}}{\partial\xi_i^2} V_{i,m}d\xi_i$$
        $$=-\frac{a_i(t)}{2}\frac{d}{dt}\int_0^1\left(\frac{\partial V_{i,m}}{\partial\xi_i}\right)^2d\xi_i=-\frac{a_i(t)}{2}\frac{d}{dt}||V_{i,m}||_{\mathbb{H}_0^1(0,1)}^2,$$
        and 
        $$\left(b(\xi_i,t)\frac{\partial V_{i,m}}{\partial \xi_i},V_{i,m}^{\prime}\right)=\int_0^1 b(\xi_i,t)\frac{\partial V_{i,m}}{\partial \xi_i}V_{i,m}^{\prime}d\xi_i=-\frac{s^{\prime}(t)}{s(t)}\int_0^1 V_{i,m}^{\prime} V_{i,m}d\xi_i$$
        $$-\int_0^1 b(\xi_i,t)\frac{\partial V_{i,m^{\prime}}}{\partial \xi_i}V_{i,m}d\xi_i=-\frac{s^{\prime}(t)}{2s(t)}\frac{d}{dt}\int_0^1 V_{i,m}^2d\xi_i=-\frac{s^{\prime}(t)}{2s(t)}\frac{d}{dt}|V_{i,m}|_{L^2(0,1)}^2,$$
        analogously as previous case 1, respect to Holder inequality we obtain
        $$\left(1-\frac{1}{2}R_i(t)\right)|V_{i,m}^{\prime}|_{L^2(0,1)}^2+\frac{a_i(t)}{2}\frac{d}{dt}||V_{i,m}||_{\mathbb{H}_0^1(0,1)}^2+\frac{s^{\prime}(t)}{2s(t)}\frac{d}{dt}|V_{i,m}|_{L^2(0,1)}^2\leq\frac{R_i(t)}{2}|F_i|_{L^2(0,1)}^2.$$
        Hence, employing integrating on $[0,t)$ with $t\in[0,T_m)$ and according to integration by part we have
        $$\int_0^t |V_{i,m}|_{L^2(0,1)}^2d\tau+||V_{i,m}(t)||_{\mathbb{H}_0^1(0,1)}^2+|V_{i,m}(t)|_{L^2(0,1)}^2$$
        $$\leq\frac{c_5}{c_4}||V_{i,m}(0)||_{\mathbb{H}_0^1(0,1)}^2+\frac{c_6}{c_4}|V_{i,m}(0)|_{L^2(0,1)}^2+\frac{c_2}{c_4}|F_i|_{L^2(0,1)}^2,$$
        where $c_5=\frac{N_i}{2}$, $c_6=\frac{s_{M^{\prime}}}{2s_m}T$ and $c_4=\min\{1-\frac{m_i}{2},\frac{N_i}{2}, \frac{s_{M^{\prime}}}{2s_m}\}$. Thus, Gronwall's Lemma guarantees that
        $$V_{i,m}\;\text{ is bounded in }\; L^{\infty}(0,T,\mathbb{H}_0^1(0,1)),\quad\quad V_{i,m}^{\prime}\;\text{ is bounded in }\; L^2(0,T,L^2(0,1)).$$
        \item[Case 3.] Assume that $\phi=-\frac{\partial^2 V_{i,m}}{\partial\xi_i^2}$, then we can rewrite equation in the system \eqref{eq62} as follows
        $$-\left(V_{i,m}^{\prime},\frac{\partial^2 V_{i,m}}{\partial\xi_i^2}\right)+a_i(t)\left(\frac{\partial^2 V_{i,m}}{\partial \xi_i^2},\frac{\partial^2 V_{i,m}}{\partial\xi_i^2}\right)+\left(b(\xi_i,t)\frac{\partial V_{i,m}}{\partial \xi_i},\frac{\partial^2 V_{i,m}}{\partial\xi_i^2}\right)$$
        $$=-R_i(t)\left(F_i(\xi_i,t),\frac{\partial^2 V_{i,m}}{\partial\xi_i^2}\right).$$
        For each term of the left-hand side by integration by parts we can provide the following evaluations
        $$\left(V_{i,m}^{\prime},\frac{\partial^2 V_{i,m}}{\partial\xi_i^2}\right)=\int_0^1 V_{i,m}^{\prime}\frac{\partial^2 V_{i,m}}{\partial\xi_i^2}d\xi_i=-\int_0^1\frac{\partial V_{i,m}^{\prime}}{\partial\xi_i}\frac{\partial V_{i,m}}{\partial\xi_i}d\xi_i=-\frac{1}{2}\frac{d}{dt}\int_0^1\left(\frac{\partial V_{i,m}}{\partial\xi_i}\right)^2d\xi_i$$
        $$=-\frac{1}{2}\frac{d}{dt}||V_{i,m}||_{\mathbb{H}_0^1(0,1)}^2,$$
        $$\left(\frac{\partial^2 V_{i,m}}{\partial \xi_i^2},\frac{\partial^2 V_{i,m}}{\partial\xi_i^2}\right)=\int_0^1\left(\frac{\partial^2V_{i,m}}{\partial\xi_i^2}\right)^2d\xi_i=\left|\frac{\partial^2V_{i,m}}{\partial\xi_i^2}\right|_{L^2(0,1)}^2,$$
        $$\left(b(\xi_i,t)\frac{\partial V_{i,m}}{\partial \xi_i},\frac{\partial^2 V_{i,m}}{\partial\xi_i^2}\right)=\int_0^1 b(\xi_i,t)\frac{\partial V_{i,m}}{\partial\xi_i}\frac{\partial^2 V_{i,m}}{\partial\xi_i^2}d\xi_i=-\frac{s^{\prime}(t)}{2s(t)}\int_0^1\left(\frac{\partial V_{i,m}}{\partial\xi_i}\right)^2d\xi_i$$
        $$=-\frac{s^{\prime}(t)}{2s(t)}||V_{i,m}||_{\mathbb{H}_0^1(0,1)}^2.$$
        Then, by Holder inequality we have
        $$\frac{1}{2}\frac{d}{dt}||V_{i,m}||_{\mathbb{H}_0^1(0,1)}^2+a_i(t)\left|\frac{\partial^2V_{i,m}}{\partial\xi_i^2}\right|_{L^2(0,1)}^2-\frac{s^{\prime}(t)}{2s(t)}||V_{i,m}||_{\mathbb{H}_0^1(0,1)}^2$$
        $$\leq\frac{1}{2}R_i(t)|F_i|_{L^2(0,1)}^2+\frac{1}{2}R_i(t)\left|\frac{\partial^2V_{i,m}}{\partial\xi_i^2}\right|_{L^2(0,1)}^2.$$
        Thus, similarly as previous cases assuming all conditions hold, we establish the following result
        $$||V_{i,m}(t)||_{\mathbb{H}_0^1(0,1)}^2+\int_0^t\left|\frac{\partial^2V_{i,m}}{\partial\xi_i^2}\right|d\tau\leq \frac{c_2}{c_7}|F_i|_{L^2(0,1)}^2+\frac{c_6}{c_7}\int_0^T||V_{i,m}||_{\mathbb{H}_0^1(0,1)}^2d\tau,$$
        where $c_2$, $c_6$ are mentioned in previous cases and $c_7=\min\{\frac{1}{2},N_i-\frac{M_i}{2}\}$. Hence, respect to Gronwall's inequality we conclude that
        $$V_{i,m}\;\text{ is bounded in }\; L^{\infty}(0,T,\mathbb{H}_0^1(0,1)),\quad\quad \frac{\partial^2V_{i,m}}{\partial\xi_i^2}\;\text{ is bounded in }\; L^2(0,T,L^2(0,1)).$$
    \end{itemize}
    \textbf{Uniqueness.} Suppose that the problem admits two global strong solution pairs
    $$(V_{i,m}^{(1)},R_i^{(1)})\quad \text{and}\quad (V_{i,m}^{(2)},R_i^{(2)})\quad i=1,2$$
    corresponding to the same initial, boundary, and overdetermination data. Let us define
    $$W_{i,m}(\xi_i,t)=V_{i,m}^{(1)}(\xi_i,t)-V_{i,m}^{(2)}(\xi_i,t)\quad\text{and}\quad \rho_i(t)=R_i^{(1)}(t)-R_i^{(2)}(t).$$
    Subtracting the two nonhomogeneous equations gives
    \begin{equation}\label{eq63}
        \frac{\partial W_{i,m}}{\partial t}-b(\xi_i,t)\frac{\partial W_{i,m}}{\partial \xi_i}-a_i(t)\frac{\partial^2 W_{i,m}}{\partial\xi_i^2}=\rho_i(t)F_i(\xi_i,t)
    \end{equation}
    and assume that heat source coefficient from overdetermination conditions can be determined in the Volterra integral form
    $$R_i(t)=G_i(t)+\int_0^tK_i(t,\tau)R_i(\tau)d\tau,\quad i=1,2$$ 
    where $G_i(t)$ and $K_i(t,\tau)$ are continuous and bounded functions in $[0,T]$, substracting $R_i^{(1)}$ and $R_i^{(2)}$, it which implies
    \begin{equation}\label{eq64}
        \rho_i(t)=\int_0^t K_i(t,\tau)\rho_i(\tau)d\tau.
    \end{equation}
    If we assume that $|K_i(t,\tau)|<C_i$ holds, then it follows that $|\rho_i(t)|\leq C_i\int_0^t |\rho_i(\tau)|d\tau.$
    By Gronwall's inequality, we obtain that $\rho_i(t)\equiv 0$ for $t\in[0,T]$. Therefore, $R_i^{(1)}(t)=R_i^{(2)}(t)$.
    
    Thus, the differenial equation is also nonhomogeneous unless $\rho_i=0$. Multiplying both sides \eqref{eq63} by $W_{i,m}$ and integrate over $(0,1)$:
    $$\int_0^1\frac{\partial W_{i,m}}{\partial t}W_{i,m}d\xi_i-\int_0^1b(\xi_i,t)\frac{\partial W_{i,m}}{\partial \xi_i} W_{i,m}d\xi_i-a_i(t)\int_0^1\frac{\partial^2 W_{i,m}}{\partial\xi_i^2}W_{i,m}d\xi_i$$
    $$=\rho_i(t)\int_0^1F_i(\xi_i,t)W_{i,m}d\xi_i.$$
    Again, for the each term of the left-hand side we have
    $$\int_0^1\frac{\partial W_{i,m}}{\partial t}W_{i,m}d\xi_i=\frac{1}{2}\frac{d}{dt}||W_{i,m}(t)||_{L^2(0,1)}^2,\quad -a_i(t)\int_0^1\frac{\partial^2 W_{i,m}}{\partial\xi_i^2}W_{i,m}d\xi_i=a_i(t)\left\|\frac{\partial W_{i,m}}{\partial\xi_i}(t)\right\|_{L^2(0,1)}^2,$$
    $$\int_0^1b(\xi_i,t)\frac{\partial W_{i,m}}{\partial \xi_i} W_{i,m}d\xi_i=\frac{1}{2}\int_0^1\frac{\partial b}{\partial \xi_i}W_{i,m}^2d\xi_i.$$
    The right-hand side is estimated by Cauchy–Schwarz and Young inequalities:
    $$|\rho_i(t)(F_i,W_{i,m})_{L^2(0,1)}|\leq |\rho_i(t)|||F_i(t)||_{L^{2}(0,1)}||W_{i,m}(t)||_{L^2(0,1)}$$
    $$\leq \frac{1}{2}|\rho_i(t)|^2||F_i||_{L^2(0,1)}^2+\frac{1}{2}||W_{i,m}(t)||_{L^2(0,1)}^2.$$
    Moreover, 
    $$\left|\int_0^1\frac{\partial b}{\partial \xi_i}|W_{i,m}|^2d\xi_i\right|\leq \left\|\frac{\partial b}{\partial \xi_i}(\cdot,t)\right\|_{L^{\infty}(0,1)}||W_{i,m}||_{L^2(0,1)}^2.$$
    Substitution all these estimations into \eqref{eq63} gives the energy identity
    $$\frac{d}{dt}||W_{i,m}(t)||_{L^2(0,1)}^2+2a_i(t)\left\|\frac{\partial W_{i,m}}{\partial\xi_i}(t)\right\|_{L^2(0,1)}^2\leq\left(1+\left\|\frac{\partial b}{\partial \xi_i}(\cdot,t)\right\|_{L^{\infty}(0,1)}\right)||W_{i,m}(t)||_{L^2(0,1)}^2$$
    $$+|\rho_i(t)|||F_i(t)||_{L^2(0,1)}^2.$$
    As it is shown that $\rho_i(t)\equiv 0$, by integrating over $[0,t)$ with $t\in [0,T)$, we can obtain the next estimation
    $$||W_{i,m}(t)||_{L^2(0,1)}^2+2N_i\int_0^T\left\|\frac{\partial W_{i,m}}{\partial\xi_i}(\tau)\right\|_{L^2(0,1)}^2d\tau \leq C_0\int_0^T ||W_{i,m}(t)||_{L^2(0,1)}^2$$,
    where $C_0=\left|1+\left\|\frac{\partial b}{\partial \xi_i}(\cdot,t)\right\|_{L^{\infty}(0,1)}\right|$. Hence, by Gronwall's inequality we can conclude that $||W_{i,m}(t)||_{L^2(0,1)}^2\equiv 0$, which leads to $W_{i,m}=0$ and $V_{i,m}^{(1)}=V_{i,m}^{(2)}$.
\end{proof}
Consequently, we obtain the next result.
\begin{theorem}\label{thm5}
Let suppose $s_0$, $s_m$, $s_{m^{\prime}}$, $s_M$ and $s_{M^{\prime}}$ are positive real numbers such that $s(t)\in C^2(0,T)$ and $s(0)=s_0$ with $0<s_m<s(t)<s_M$, $s_{m^{\prime}}<s^{\prime}(t)<s_{M^{\prime}}$ and $s^{\prime}(t)\in L^{\infty}(0,T)$ for all $t\in[0,T]$. Moreover, initial datas $\varphi_{i0}\in \mathbb{H}_0^1(0,s_0)$, heat sources $f_1(x,t)\in L^1(0,
T,L^2(0,s(t)))\cap L^2(0,T,L^2(0,s(t)))$ and $f_2(x,t)\in L^1(0,
T,L^2(s(t),d))\cap L^2(0,T,L^2(s(t),d))$ given, for the uniquely reconstructed coefficients $R_i(t)$ obtained in Theorem 1, 2, and Theorem 3, problem \eqref{eq10} possesses a unique strong solution pair $(u_1,u_2)$ of the problem \eqref{eq4}, that is, 
$$\frac{\partial u_1}{\partial t}-a_1\frac{\partial^2 u_1}{\partial x^2}=R_1(t)f_1(x,t),\quad\text{in}\quad L^2(0,T,L^2(0,s(t))),$$
$$\frac{\partial u_2}{\partial t}-a_2\frac{\partial^2 u_2}{\partial x^2}=R_2(t)f_2(x,t),\quad\text{in}\quad L^2(0,T,L^2(s(t),d)),$$
which satisfy the following conditions:
\begin{itemize}
    \item[a)] $u_1\in L^{\infty}(0,T; \mathbb{H}_0^1(0,s(t))\cap \mathbb{H}^2(0,s(t)))$ and $u_2\in L^{\infty}(0,T; H_0^1(s(t),d)\cap \mathbb{H}^2(s(t),d))$,
    \item[b)] $\partial_t u_1\in L^2(0,T; L^2(0,s(t)))$ and $\partial_t u_2\in L^2(0,T; L^2(s(t),d))$,
    \item[c)] $R_i(t)\in C[0,T]$, $R_i(t)>0$ with $i=1,2$ for all $t\geq 0$ and $0<m_i<R_i(t)<M_i$ with $m_i,M_i>0$.  
\end{itemize}
    
\end{theorem}

\section{The reaction coefficient Stefan type problem with two moving boundaries}
In this section, we introduce the new domain $\Omega_3=(s(t),r(t))\times(0,T]$ and consider the inverse two-phase Stefan type problem with Neumann condition and two moving boundaries, determining the functions $u_1(x,t)$, $u_2(x,t)$ and restoring the time-dependent coefficient of the reaction terms $P_1(t)$, $P_2(t)$ in each phase such that
\begin{equation}\label{eq65}
    \begin{cases}
        \frac{\partial u_1}{\partial t}=a_1^2\frac{\partial^2 u_1}{\partial x^2}+P_1(t)u_1(x,t)+f_1(x,t),&(x,t)\in\Omega_1,\\
        \frac{\partial u_2}{\partial t}=a_2^2\frac{\partial^2 u_2}{\partial x^2}+P_2(t)u_2(x,t)+f_2(x,t),&(x,t)\in\Omega_3,\\
        u_1(x,0)=\varphi_1(x),\quad s(0)=s_0,&\forall x\in[0,s(0)],\\
        u_2(x,0)=\varphi_2(x),\quad r(0)=r_0,&\forall x\in[s(0),r(0)],\\
        -k_1\frac{\partial u_1}{\partial x}(0,t)=g_1(t),&\forall t\in[0,T],\\
        u_1(s(t),t)=u_2(s(t),t)=u^*,&\forall t\in[0,T],\\
        -k_1\frac{\partial u_1}{\partial x}(s(t),t)=-k_2\frac{\partial u_2}{\partial x}(s(t),t)+Ls^{\prime}(t),&\forall t\in[0,T],\\
        u_2(r(t),t)=g_2(t),&\forall t\in[0,T],\\
        -k_2\frac{\partial u_2}{\partial x}(r(t),t)=Lr^{\prime}(t),&\forall t\in[0,T],
    \end{cases}
\end{equation}
where $f_i(x,t)$, $g_i(t)$ and $\varphi_i(x)$ with $i=1,2$ are given functions, $a_i^2$, $k_i$ with $i=1,2$, $u^*$ and $L$ are known fixed constants, and $s_0, r_0$ are positive given real numbers.

One advantage of problem \eqref{eq65}, compared with problem \eqref{eq1} in the previous section, is that no additional overdetermination condition is required to determine the time-dependent coefficients. In the present formulation, the two Stefan conditions, namely the seventh and ninth conditions in \eqref{eq65}, provide the necessary relations for estimating both coefficients.

Now we introduce the following substitutions which help us to reduce the problem \eqref{eq65} to the similar form of the problem \eqref{eq1}, which are
\begin{equation}\label{eq66}
    v_1(x,t)=R_1(t)u_1(x,t),\quad\quad R_1(t)=\exp\left(-\int_0^t P_1(\tau)d\tau\right),
\end{equation}
\begin{equation}\label{eq67}
    v_2(x,t)=R_2(t)u_2(x,t),\quad\quad R_2(t)=\exp\left(-\int_0^t P_2(\tau)d\tau\right).
\end{equation}
Then problem \eqref{eq65} takes the form
\begin{equation}\label{eq68}
    \begin{cases}
        \frac{\partial v_1}{\partial t}=a_1^2\frac{\partial^2 v_1}{\partial x^2}+R_1(t)f_1(x,t),& (x,t)\in\Omega_1,\\
        \frac{\partial v_2}{\partial t}=a_2^2\frac{\partial^2 v_2}{\partial x^2}+R_2(t)f_2(x,t),& (x,t)\in\Omega_3,\\
        v_1(x,0)=\varphi_1(x),\quad s(0)=s_0,& \forall x\in[0,s(0)],\\
        v_2(x,0)=\varphi_2(x),\quad r(0)=r_0,&\forall x\in[s(0),r(0)],\\
        -k_1\frac{\partial v_1}{\partial x}(0,t)=R_1(t)g_1(t),&\forall t\in[0,T],\\
        v_1(s(t),t)=R_1(t)u^*,&\forall t\in[0,T],\\
        v_2(s(t),t)=R_2(t)u^*,&\forall t\in[0,T],\\
        -\frac{k_1}{R_1(t)}\frac{\partial v_1}{\partial x}(s(t),t)=-\frac{k_2}{R_2(t)}\frac{\partial v_2}{\partial x}(s(t),t)+Ls^{\prime}(t),&\forall t\in[0,T],\\
        v_2(r(t),t)=R_2(t)g_2(t),&\forall t\in[0,T],\\
        -k_2\frac{\partial v_2}{\partial x}(r(t),t)=LR_2(t)r^{\prime}(t),&\forall t\in[0,T],
    \end{cases}
\end{equation}
where the pairs $(u_1, R_1)$ and $(u_2, R_2)$ have to be determined, then solution pair $(u_1,P_1)$ and $(u_2,P_2)$ to the considered original problem \eqref{eq65} can be estimated in the following form
\begin{equation}\label{eq69}
    u_i(x,t)=\frac{v_i(x,t)}{R_i(t)},\quad \quad P_i(t)=-\frac{R_i^{\prime}(t)}{R_i(t)},\quad\quad i=1,2.
\end{equation}
Next, we apply the following change of variable for the purpose of the spectral approach to solve the original problem \eqref{eq65} such that
\begin{equation}\label{eq70}
    y_1(x,t)=v_1(x,t)-R_1(t)u^*+\frac{1}{k_1}R_1(t)g_1(t)(x-s(t)),
\end{equation}
\begin{equation}\label{eq71}
    y_2(x,t)=v_2(x,t)-\frac{R_2(t)u^*}{s(t)-r(t)}(x-r(t))-\frac{R_2(t)g_2(t)}{r(t)-s(t)}(x-s(t)).
\end{equation}
Hence, we present the problem \eqref{eq65} as the following form
\begin{equation}\label{eq72}
    \begin{cases}
        \frac{\partial y_1}{\partial t}=a_1^2\frac{\partial^2 y_1}{\partial x^2}+R_1(t)h_1(x,t),&(x,t)\in\Omega_1,\\
        \frac{\partial y_2}{\partial t}=a_2^2\frac{\partial^2 y_2}{\partial x^2}+R_2(t)h_2(x,t),&(x,t)\in\Omega_3,\\
        y_1(x,0)=\psi_1(x),\quad s(0)=s_0,&\forall x\in[0,s(0)],\\
        y_2(x,0)=\psi_2(x),\quad r(0)=r_0,&\forall x\in[s(0),r(0)],\\
        -k_1\frac{\partial y_1}{\partial x}(0,t)=0,&\forall t\in[0,T],\\
        y_1(s(t),t)=y_2(s(t),t)=y_2(r(t),t)=0,&\forall t\in[0,T],\\
        -\frac{k_1}{R_1(t)}\frac{\partial y_1}{\partial x}(s(t),t)=-\frac{k_2}{R_2(t)}\frac{\partial y_2}{\partial x}(s(t),t)+\frac{k_2(g_2(t)-u^*)}{s(t)-r(t)}-g_1(t)+Ls^{\prime}(t),&\forall t\in[0,T],\\
        -k_2\frac{\partial y_2}{\partial x}(r(t),t)=R_2(t)\left(\frac{k_2(u^*-g_2(t))}{s(t)-r(t)}+Lr^{\prime}(t)\right),&\forall t\in[0,T],
    \end{cases}
\end{equation}
where
\begin{equation}\label{eq73}
    h_1(x,t)=f_1(x,t)+\frac{1}{k_1}\left(\frac{R_1^{\prime}(t)}{R_1(t)}g_1(t)+g_1^{\prime}(t)\right)(x-s(t))-\frac{R_1^{\prime}(t)u^*}{R_1(t)}-\frac{1}{k_1}g_1(t)s^{\prime}(t),
\end{equation}
\begin{equation}\label{eq74}
    \begin{split}
        h_2(x,t)&=f_2(x,t)-\frac{R_2^{\prime}(t)u^*(s(t)-r(t))-R_2(t)u^*(s^{\prime}(t)-r^{\prime}(t))}{R_2(t)(s(t)-r(t))^2}(x-r(t))\\
        &-\frac{(R_2^{\prime}(t)g_2(t)+R_2(t)g_2^{\prime}(t))(r(t-s(t)))-R_2(t)g_2(t)(r^{\prime}(t)-s^{\prime}(t))}{R_2(t)(r(t)-s(t))^2}(x-s(t))\\
        &+\frac{u^*r^{\prime}(t)}{s(t)-r(t)}+\frac{g_2(t)s^{\prime}(t)}{r(t)-s(t)},
    \end{split}
\end{equation}
\begin{equation}\label{eq75}
    \psi_1(x)=\varphi_1(x)-u^*+\frac{1}{k_1}g_1(0)(x-s_0),\quad\quad g_1(0)>0,
\end{equation}
\begin{equation}\label{eq76}
    \psi_2(x)=\varphi_2(x)-\frac{u^*(x-r_0)}{s_0-r_0}-\frac{g_2(0)(x-s_0)}{r_0-s_0},\quad\quad g_2(0)>0.
\end{equation}
The next step for spectral approach we need to study the problem in the fixed domain $\Omega=(0,1)\times(0,T]$, then in the system of mapping \eqref{eq8} and \eqref{eq9}, only \eqref{eq9} can be rewritten as follows
\begin{equation}\label{eq77}
    V_2:\;\Omega_2\rightarrow\Omega,\quad\quad (x,t)\mapsto (\xi_2,t)=\left(\frac{x-s(t)}{r(t)-s(t)},t\right),
\end{equation}
where function $V_2\in C^{2,1}(\Omega)$ and $V_2^{-1}\in C^{2,1}(\Omega_2)$. Therefore, the problem \eqref{eq72} can be reduced to the form
\begin{equation}\label{eq78}
    \begin{cases}
        \frac{\partial V_1}{\partial t}-b(\xi_1,t)\frac{\partial V_1}{\partial \xi_1}-a_1(t)\frac{\partial^2 V_1}{\partial \xi_1^2}=R_1(t)F_1(\xi_1,t), &(\xi_1,t)\in\Omega,\\
        \frac{\partial V_2}{\partial t}-b(\xi_2,t)\frac{\partial V_2}{\partial \xi_3}-a_2(t)\frac{\partial^2 V_2}{\partial \xi_3^2}=R_2(t)F_2(\xi_2,t), &(\xi_2,t)\in\Omega,\\
        V_1(\xi_1,0)=\Psi_1(\xi_1),&\forall \xi_1\in[0,1],\\
        V_2(\xi_2,0)=\Phi_2(\xi_2),&\forall \xi_2\in[0,1],\\
        \frac{\partial V_1}{\partial \xi_1}(0,t)=0, & \forall t\in[0,T],\\
        V_1(1,t)=V_2(0,t)=V_2(1,t)=0,&\forall t\in[0,T],\\
        -\frac{k_1}{R_1(t)s(t)}\frac{\partial V_1}{\partial\xi_1}(1,t)=-\frac{k_2}{R_2(t)(r(t)-s(t))}\frac{\partial V_2}{\partial \xi_2}(0,t)+\frac{k_2(g_2(t)-u^*)}{s(t)-r(t)}-g_1(t)+Ls^{\prime}(t),&\forall t\in[0,T],\\
        -k_2\frac{\partial V_2}{\partial\xi_2}(1,t)=R_2(t)(r(t)-s(t))\left(\frac{k_2(u^*-g_2(t))}{s(t)-r(t)}+Lr^{\prime}(t)\right),&\forall t\in[0,T],
    \end{cases}
\end{equation}
where $a_1(t)=\frac{a_1^2}{s^2(t)}$, $a_2(t)=\frac{a_2^2}{(r(t)-s(t))^2}$, $b(\xi_i,t)=\frac{\xi_1 s^{\prime}(t)}{s(t)}$, $b(\xi_2,t)=\frac{s^{\prime}(t)(r(t)-s(t))+(x-s(t))(r^{\prime}(t)-s^{\prime}(t))}{(r(t)-s(t))^2}$, $F_1(\xi_1,t)=h_1(\xi_1s(t),t)$, $F_2(\xi_2,t)=h_2(\xi_2(r(t)-s(t))+s(t),t)$, $\Psi_1(\xi_1)=\psi_1(\xi_1s(t))$ and $\Phi_2(\xi_2)=\psi_2(\xi_2(r(t)-s(t))+s(t))$ where $h_i$, $\psi_i$ with $i=1,2$ are defined by \eqref{eq73},\eqref{eq74},\eqref{eq75} and \eqref{eq76}.

Next, we consider the spectral problem for the first phase $V_1(\xi_!,t)$:
\begin{equation}\label{eq79}
    \begin{cases}
        -w_{1,n}^{\prime\prime}(\xi_1)=\widetilde{\lambda}_n w_{1,n}(\xi_1),&0<\xi_1<1,\\
        w_{1,n}^{\prime}(0)=w_{1,n}(1)=0,
    \end{cases}
\end{equation}
which gives the following eigenvalue and orthonormal eigenfunction 
\begin{equation}\label{eq80}
    w_{1,n}(\xi_1)=\sqrt{2}\cos\left(\sqrt{\widetilde{\lambda}_n}\xi_1\right),\quad i=1,2,\quad\quad \widetilde{\lambda}_n=\frac{(2n-1)^2\pi^2}{4},\quad\quad n=1,2,3,...\;.
\end{equation}
For the second phase $V_2(\xi_2,t)$, we consider the same spectral problem \eqref{eq11} as in previous section which gives the orthonormal eigenfunction $\phi_{2,n}(\xi_2)$ and eigenvalue $\lambda_n$ defined by \eqref{eq12}.

Then assuming that $\{w_{1,n}(\xi_1)\}_{n=1}^{\infty}$ and $\{\phi_{2,n}(\xi_2)\}_{n=1}^{\infty}$ be orthonormal sets on $L^2(0,1)$, then analogously as in previous Section 1, according to \eqref{eq19} we set that weak solution of the problem \eqref{eq78} can be presented as
\begin{equation}\label{eq81}
    V_{1,n}(\xi_1)=\sum\limits_{n=1}^{\infty}\Bigg(\Psi_{1,n}e^{-\int_0^t(a_1(\tau)\widetilde{\lambda}_n-b_1(\tau))d\tau}+\int\limits_0^t R_1(\tau)F_{1,n}(\tau)e^{-\int_{\tau}^t(a_1(z)\widetilde{\lambda}_n-b_1(z))dz}d\tau\Bigg)w_{1,n}(\xi_1),
\end{equation}
\begin{equation}\label{eq82}
    V_{2,n}(\xi_2)=\sum\limits_{n=1}^{\infty}\Bigg(\Phi_{2,n}e^{-\int_0^t(a_2(\tau)\lambda_n-b_2(\tau))d\tau}+\int\limits_0^t R_2(\tau)F_{2,n}(\tau)e^{-\int_{\tau}^t(a_2(z)\lambda_n-b_2(z))dz}d\tau\Bigg)\phi_{2,n}(\xi_2)
\end{equation}
where $b_1(t)=(b(\xi_1,t)w_{1,n}^{\prime},w_{1,n})_{L^2(0,1)}$, $b_2(t)=(b(\xi_2,t)\phi_{2,n}^{\prime},\phi_{2,n})_{L^2(0,1)}$, $\Psi_{1,n}=(\Psi_1(\xi_1),w_{1,n}(\xi_1))_{L^2(0,1)}$, $\Phi_{2,n}=(\Phi_2(\xi_2),\phi_{2,n}(\xi_2))_{L^2(0,1)}$, $F_{1,n}(t)=(F_{1}(\xi_1,t),w_{1,n}(\xi_1))$ and $F_{2,n}(t)=(F_{2}(\xi_2,t),\phi_{2,n}(\xi_2))$, where $w_{1,n}$ is defined by \eqref{eq80} and $\phi_{2,n}$ is given by \eqref{eq12}. 

From the system of the last two conditions in the system \eqref{eq78} one can determine the time-dependent coefficients $R_1(t)$ and $R_2(t)$ dependent to each others, such that
\begin{equation}\label{eq84}
    \begin{split}
        R_1(t)&=\frac{\sqrt{2}k_1}{s(t)}\sum_{n=1}^{\infty}\Bigg((-1)^n\sqrt{\widetilde{\lambda}_n}\Psi_{1,n}e^{-\int_0^t(a_1(\tau)\widetilde{\lambda}_n-b_1(\tau))d\tau}\\
        &+\int_0^t R_1(\tau)(-1)^n\sqrt{\widetilde{\lambda}_n}F_{1,n}(\tau)e^{-\int_{\tau}^t(a_1(z)\widetilde{\lambda}_n-b_1(z))dz}d\tau\Bigg)\\
        &\times\Bigg[-\frac{\sqrt{2}k_2}{R_2(t)(r(t)-s(t))}\sum_{n=1}^{\infty}\Bigg(\sqrt{\lambda}_n\Phi_{2,n}e^{-\int_0^t(a_2(\tau)\lambda_n-b_2(\tau))d\tau}\\
        &+\int_0^t R_2(\tau)\sqrt{\lambda_n}F_{2,n}(\tau)e^{-\int_{\tau}^t(a_2(z)\lambda_n-b_2(z))dz}d\tau\Bigg)+\frac{k_2(g_2(t)-u^*)}{s(t)-r(t)}-g_1(t)+Ls^{\prime}(t)\Bigg]^{-1}
    \end{split}
\end{equation}
where
\begin{equation}\label{eq85}
    \begin{split}
        R_2(t)&=-\frac{\sqrt{2}k_2}{Lr^{\prime}(t)(r(t)-s(t))+k_2(g_2(t)-u^*)}\sum\limits_{n=1}^{\infty}\Bigg((-1)^n\sqrt{\lambda_n}\Phi_{2,n}e^{-\int_0^t(a_2(\tau)\lambda_n-b_2(\tau))d\tau}\\
        &+\int_0^t R_2(\tau)(-1)^n\sqrt{\lambda_n}F_{2,n}(\tau)e^{-\int_{\tau}^t(a_2(z)\lambda_n-b_2(z))dz}d\tau\Bigg).
    \end{split}
\end{equation}
The proof of the existence and uniqueness of the solutions \eqref{eq81}, \eqref{eq82}, \eqref{eq84} and \eqref{eq85} of the problem \eqref{eq78} are similar as in previous Section 1 which can be established by the following results.
\begin{theorem}\label{thm6}
    Assume the inequalities in Lemma \ref{lem1}, Lemma \ref{lem2} and (i)-(iv) in Theorem \ref{thm1} are hold. Moreover, if the following conditions
    \begin{itemize}
        \item[(H1)] $r(t)\in C^2[0,T]$ such that $r(0)=r_0$, $r(t)>s(t)$ for $\forall t\in [0,T]$ and $0<r_m<r(t)<r_M$, $0<r_{m'}<r'(t)<r_{M'}$ for $r_M,r_{M'}$ positive constants; 
        \item[(H2)] $r^{\prime}(t)\in L^{\infty}(0,T)$
    \end{itemize}
    are satisfied, then there exists a unique weak solution pairs $(V_i, R_i)$ where $i=1,2$ to the problem \eqref{eq78}.
\end{theorem}
\begin{proof}
    It can be implemented analogously as in study \cite{20}.
\end{proof}

\begin{definition}
A pair of functions $\bigl((u_1,P_1),(u_2,P_2)\bigr)$ is called a weak solution of problem \eqref{eq66} if there exist functions $R_i\in C^1([0,T]),\; R_i(t)>0,\; R_i(0)=1,\; i=1,2,$ such that $P_i(t)=-\frac{R_i'(t)}{R_i(t)},\; i=1,2$ and the transformed functions $ v_i(x,t)=R_i(t)u_i(x,t),\;i=1,2,$ after the transformations
\[
y_1(x,t)=v_1(x,t)-R_1(t)u^*+\frac{1}{k_1}R_1(t)g_1(t)(x-s(t)),
\]
\[
y_2(x,t)=v_2(x,t)-\frac{R_2(t)u^*}{s(t)-r(t)}(x-r(t))
-\frac{R_2(t)g_2(t)}{r(t)-s(t)}(x-s(t)),
\]
and
\[
V_1(\xi_1,t)=y_1(s(t)\xi_1,t),\quad V_2(\xi_2,t)=y_2\bigl(s(t)+(r(t)-s(t))\xi_2,t\bigr),
\]
satisfy
\[
V_i\in C([0,T];L^2(0,1))
\cap C((0,T];\mathbb{H}_0^1 (0,1))
\cap C^1((0,T];L^2(0,1)),
\qquad i=1,2.
\]
Finally, the original temperature functions $$u_1\in C([0,T],L^2(0,s(t)))\cap C((0,T];\mathbb{H}_0^2(0,s(t)))\cap C^1((0,T];L^2(0,s(t))),$$ 
$$u_2\in C([0,T],L^2(s(t),r(t)))\cap C((0,T];\mathbb{H}_0^2(s(t),r(t)))\cap C^1((0,T];L^2(s(t),r(t)))$$ are recovered by
\[
u_1(x,t)
=
\frac{1}{R_1(t)}
V_1\left(\frac{x}{s(t)},t\right)
+
u^*
-
\frac{g_1(t)}{k_1}(x-s(t)),
\qquad 0<x<s(t),
\]
and
\[
u_2(x,t)
=
\frac{1}{R_2(t)}
V_2\left(\frac{x-s(t)}{r(t)-s(t)},t\right)
+
\frac{u^*(x-r(t))}{s(t)-r(t)}
+
\frac{g_2(t)(x-s(t))}{r(t)-s(t)},
\qquad s(t)<x<r(t).
\]
\end{definition}

\textbf{Example 2.} Let us consider the inverse Stefan type two-phase problem with suitable source functions and initial functions, and the domain with two moving boundaries  $\Omega_1=(0,\pi\sqrt{t+1})\times(0,T]$ and $\Omega_2=(\pi\sqrt{t+1},2\pi\sqrt{t+1})\times(0,T]$ in which the pairs $(u_1(x,t),P_1(t))$ and $(u_2(x,t),P_2(t))$ have to be determined, such that
\begin{equation}\label{eq89}
    \begin{cases}
        \frac{\partial u_1}{\partial t}=a_1^2\frac{\partial^2 u_1}{\partial x^2}+P_1(t)u_1(x,t)+\sqrt{2}(t+1)^{-\frac{a_1^2+1}{4}}\cos\left(\frac{x}{2\sqrt{t+1}}\right),&(x,t)\in\Omega_1,\\
        \frac{\partial u_2}{\partial t}=a_2^2\frac{\partial^2 u_2}{\partial x^2}+P_2(t)u_2(x,t)+\sqrt{2}(t+1)^{-(a_2^2+1/4)}\sin\left(\frac{x-\pi\sqrt{t+1}}{\sqrt{t+1}}\right),&(x,t)\in\Omega_2,\\
        u_1(x,0)=\sqrt{2}\cos\left(\frac{x}{2}\right),&\forall x\in[0,\pi],\\
        u_2(x,0)=\sqrt{2}\sin\left(x-\pi\right),&\forall x\in[\pi,2\pi],\\
        -k_1\frac{\partial u_1}{\partial x}(0,t)=0,&\forall t\in[0,T],\\
        u_1(\pi\sqrt{t+1},t)=u_2(\pi\sqrt{t+1},t)=0,&\forall t\in[0,T],\\
        -k_1\frac{\partial u_1}{\partial x}(\pi\sqrt{t+1},t)=-k_2\frac{\partial u_2}{\partial x}(\pi\sqrt{t+1},t)+\frac{L\pi}{2\sqrt{t+1}},&\forall t\in[0,T],\\
        u_2(2\pi\sqrt{t+1},t)=0,&\forall t\in[0,T],\\
        -k_2\frac{\partial u_2}{\partial x}(2\pi\sqrt{t+1},t)=\frac{L\pi}{\sqrt{t+1}},&\forall t\in[0,T].
    \end{cases}
\end{equation}
Applying the substitutions \eqref{eq66} and \eqref{eq67}, the problem \eqref{eq89} can be reduced to the inverse problem where the pairs $(v_1(x,t),R_1(t))$ and $(v_2(x,t),R_2(t))$ have to be restored which can be presented in the following form
\begin{equation}\label{eq90}
    \begin{cases}
        \frac{\partial v_1}{\partial t}=a_1^2\frac{\partial^2 v_1}{\partial x^2}+R_1(t)\sqrt{2}(t+1)^{-\frac{a_1^2+1}{4}}\cos\left(\frac{x}{2\sqrt{t+1}}\right),&(x,t)\in\Omega_1,\\
        \frac{\partial v_2}{\partial t}=a_2^2\frac{\partial^2 v_2}{\partial x^2}+R_2(t)\sqrt{2}(t+1)^{-(a_2^2+1/4)}\sin\left(\frac{x-\pi\sqrt{t+1}}{\sqrt{t+1}}\right),&(x,t)\in\Omega_2,\\
        v_1(x,0)=\sqrt{2}\cos\left(\frac{x}{2}\right),&\forall x\in[0,\pi],\\
        v_2(x,0)=\sqrt{2}\sin\left(x-\pi\right),&\forall x\in[\pi,2\pi],\\
        -k_1\frac{\partial v_1}{\partial x}(0,t)=0,&\forall t\in[0,T],\\
        v_1(\pi\sqrt{t+1},t)=v_2(\pi\sqrt{t+1},t)=0,&\forall t\in[0,T],\\
        -\frac{k_1}{R_1(t)}\frac{\partial v_1}{\partial x}(\pi\sqrt{t+1},t)=-\frac{k_2}{R_2(t)}\frac{\partial u_2}{\partial x}(\pi\sqrt{t+1},t)+\frac{L\pi}{2\sqrt{t+1}},&\forall t\in[0,T],\\
        v_2(2\pi\sqrt{t+1},t)=0,&\forall t\in[0,T],\\
        -k_2\frac{\partial v_2}{\partial x}(2\pi\sqrt{t+1},t)=R_2(t)\frac{L\pi}{\sqrt{t+1}}.&\forall t\in[0,T].
    \end{cases}
\end{equation}
Next, introducing the fixed domain $\Omega=(0,1)\times(0,T]$ and employing the change of variables
\begin{equation}\label{eq91}
    v_1(x,t)=V_1(\xi_1,t),\quad\quad \xi_1=\frac{x}{\pi(t+1)},\quad\quad (\xi_1,t)\in\Omega,
\end{equation}
\begin{equation}\label{eq92}
    v_2(x,t)=V_2(\xi_2,t),\quad\quad \xi_2=\frac{x-\pi(t+1)}{\pi(t+1)},\quad\quad (\xi_2,t)\in\Omega,
\end{equation}
then problem \eqref{eq90} takes the form
\begin{equation}\label{eq90-a}
    \begin{cases}
        \frac{\partial V_1}{\partial t}-\frac{\xi_1}{2(t+1)}\frac{\partial V_1}{\partial \xi_1}-\frac{a_1^2}{\pi^2(t+1)}\frac{\partial^2V_1}{\partial\xi_1^2}=R_1(t)\sqrt{2}(t+1)^{-\frac{a_1^2+1}{4}}\cos\left(\frac{\pi}{2}\xi_1\right),&(\xi_1,t)\in\Omega,\\
        \frac{\partial V_2}{\partial t}-\frac{1+\xi_2}{t+1}\frac{\partial V_2}{\partial \xi_2}-\frac{a_2^2}{\pi^2(t+1)^2}\frac{\partial^2V_2}{\partial\xi_2^2}=R_2(t)\sqrt{2}(t+1)^{-(a_2^2+1/4)}\sin\left(\pi\xi_2\right),&(\xi_2,t)\in\Omega,\\
        V_1(\xi_1,0)=\sqrt{2}\cos\left(\frac{\pi}{2}\xi_1\right),&\forall\xi_1\in[0,1],\\
        V_2(\xi_2,0)=\sqrt{2}\sin\left(\pi\xi_2\right),&\forall\xi_2\in[0,1],\\
        -k_1\frac{\partial V_1}{\partial\xi_1}(0,t)=0,&\forall t\in[0,T],\\
        V_1(1,t)=V_2(0,t)=V_2(1,t)=0,&\forall t\in[0,T],\\
        -\frac{k_1}{R_1(t)}\frac{\partial V_1}{\partial\xi_1}(1,t)=-\frac{k_2}{R_2(t)}\frac{\partial V_2}{\partial\xi_2}(0,t)+\frac{L\pi^2}{2},&\forall t\in[0,T],\\
        -k_2\frac{\partial V_2}{\partial\xi_2}(1,t)=R_2(t)L\pi^2,&\forall t\in[0,T].
    \end{cases}
\end{equation}
According to the spectral problems \eqref{eq11} and  \eqref{eq79} the orthonormal eigenfunctions are $w_{1,n}(\xi_1)=\sqrt{2}\cos\left(\frac{(2n-1)\pi}{2}\xi_1\right)$ and $\phi_{2,n}(\xi_2)=\sqrt{2}\sin\left(n\pi\xi_2\right)$ where $n\in\mathbb{N}$. Therefore, we get $(\Psi_{1}(\xi_1),w_{1,n}(\xi_1))_{L^2(0,1)}=0$, $(\Phi_{2}(\xi_2),\phi_{2,n}(\xi))_{L^2(0,1)}=0$ when $n\neq 0$. Thus, $\Psi_{1,1}=1$ and $\Psi_{2,1}=1$ when $n=1$. Analogously, $F_{1,n}(t)=F_{2,n}(t)=0$ when $n\neq 1$, otherwise $F_{1,1}(t)=(t+1)^{-\frac{a_1^2+1}{4}}$ and $F_{2,1}(t)=(t+1)^{-(a_2^2+1/4)}$. Hence, the solution functions $V_1(\xi_1,t)$ and $V_2(\xi_2,t)$ takes the form
\begin{equation}\label{eq91}
    V_{1}(\xi_1,t)=V_{1,1}(t)w_{1,1}(\xi_1),\quad\quad V_2(\xi_2,t)=V_{2,1}(t)\phi_{2,1}(\xi_2).
\end{equation}
Taking into account $\left(\frac{\xi_1}{2(t+1)w_{1,1}^{\prime},w_{1,1}}\right)_{L^2(0,1)}=\left(\frac{1+\xi_2}{2(t+1)}\phi_{1,1}^{\prime},\phi_{1,1}\right)_{L^2(0,1)}=-\frac{1}{4(t+1)}$, then $V_{1,1}(t)$ and $V_{2,1}(t)$ are the solution of the problems governing with linear differential equations such that
\begin{equation}\label{eq92}
    \begin{cases}
        V_{1,1}^{\prime}(t)+\frac{a_1^2+1}{4(t+1)}V_{1,1}(t)=R_1(t)(t+1)^{-\frac{a_1^2+1}{4}},\\
        V_{1,1}(0)=1,
    \end{cases}
\end{equation}
\begin{equation}\label{eq93}
    \begin{cases}
        V_{2,1}^{\prime}(t)+\frac{a_2^2+1/4}{t+1}V_{2,1}(t)=R_2(t)(t+1)^{-(a_2^2+1/4)},\\
        V_{2,1}(0)=1.
    \end{cases}
\end{equation}
Thus, from \eqref{eq91} we can obtain the exact solution of the problem \eqref{eq90-a} that are
\begin{equation}\label{eq94}
    V_{1}(\xi_1,t)=\sqrt{2}(t+1)^{-\frac{a_1^2+1}{4}}\left(1+\int_0^t R_1(\tau)d\tau\right)\cos\left(\frac{\pi}{2}\xi_1\right)
\end{equation}
\begin{equation}\label{eq95}
    V_2(\xi_2,t)=\sqrt{2}(t+1)^{-\left(a_2^2+\frac{1}{4}\right)}\left(1+\int_0^t R_2(\tau)d\tau\right)\sin(\pi\xi_2).
\end{equation}
Then, from the system of two conditions in \eqref{eq90-a}, it can be estimated that the time-dependent coefficients which independent to each others and defined by Volterra integral equations, such that
\begin{equation}\label{eq96}
    R_1(t)=\frac{\sqrt{2}k_1}{L\pi}(t+1)^{-\frac{a_1^2+1}{4}}\left(1+\int_0^t R_1(\tau)d\tau\right),
\end{equation}
\begin{equation}\label{eq97}
    R_2(t)=\frac{\sqrt{2}k_2}{L\pi}(t+1)^{-(a_2^2+1/4)}\left(1+\int_0^t R_2(\tau)d\tau\right).
\end{equation}
Hence, the explicit solutions of the integral equations \eqref{eq96} and \eqref{eq97} are
\begin{equation}\label{eq98}
    R_1(t)=\frac{\sqrt{2}k_1}{L\pi}(t+1)^{-\frac{a_1^2+1}{4}}\exp\left(\frac{4\sqrt{2}k_1}{L\pi(3-a_1^2)}\left[(t+1)^{\frac{3-a_1^2}{4}}-1\right]\right),
\end{equation}
\begin{equation}\label{eq99}
    R_2(t)=\frac{\sqrt{2}k_2}{L\pi}(t+1)^{-\left(a_2^2+\frac{1}{4}\right)}\exp\left(\frac{4\sqrt{2}k_2}{L\pi(3-4a_2^2)}\left[(t+1)^{\frac{3-4a_2^2}{4}}-1\right]\right).
\end{equation}
By making backward substitution to \eqref{eq69}, the solution of the original problem \eqref{eq89} can be presented as
\begin{equation}\label{eq98}
    u_1(x,t)=\frac{1}{R_1(t)}V_1\left(\frac{x}{\pi(t+1)},t\right),\quad\quad u_2(x,t)=\frac{1}{R_2(t)}V_2\left(\frac{x-\pi(t+1)}{\pi(t+1)},t\right)
\end{equation}
where $V_1$ and $V_2$ defined by \eqref{eq94} and \eqref{eq95}, and applying \eqref{eq69} we can determine coefficients $P_i(t)=-\frac{R_i^{\prime}(t)}{R_i(t)}$ with $i=1,2$, where $R_1(t)$ and $R_2(t)$ are given by \eqref{eq98} and \eqref{eq99}.

To construct the graphical formulation of the time-dependent coefficients with noisy data measurement for testing solution, we attempt to follow the steps:

\begin{enumerate}
\item To check the time-depended coefficients' sensitivity to noisy data measurement, we add relative noise level to the Stefan condition
$$-k_2\frac{\partial u_2^{\delta}}{\partial x}(r(t),t)=Lr^{\prime}(t)+\mathcal{O}(0,\sigma^2),$$
where $r(t)=2\pi\sqrt{t+1}$ and $\mathcal{O}(0,\sigma^2)$ represents the normal distribution with mean zero and standard deviation $\sigma=\delta\times\max\limits_{(r(t),t),t\in[0,T]}\bigg|\frac{\partial u_2}{\partial x}(r(t),t)\bigg|$, and $\delta$ denotes relative (in percentage) noise levels. 

\item Fix the final time $T>0$ and the number of time steps $N\geq 2$. Define $h=\frac{T}{N}$ and $t_n=nh,\; n=0,1,\ldots,N.$

\item Setting $k_1=k_2=\frac{\pi}{\sqrt{2}}$ and  $a_1=a_2=L=1$ we get $\max\limits_{(2\pi\sqrt{t+1},t),t\in[0,T]}\bigg|\frac{\partial u_2}{\partial x}(2\pi\sqrt{t+1},t)\bigg|=\sqrt{\frac{2}{T+1}}$, which implies $\sigma=\delta\sqrt{\frac{2}{T+1}}$, then we need to define $$A_1^{\delta}(t)=\frac{k_1\pi}{\sqrt{2}}\left(\frac{L\pi^2}{2}-\pi\sqrt{t+1}\mathcal{O}(0,\sigma^2)\right)(t+1)^{-\frac{a_1^2+1}{4}},$$
$$A_2^{\delta}(t)=\frac{\sqrt{2}k_2\pi}{L\pi^2+\pi\sqrt{t+1}\mathcal{O}(0,\sigma^2)}(t+1)^{-(a_2^2+\frac{1}{4})}$$

\item Introduce the auxiliary functions $Y_i(t)=1+\int_0^t R_i(\tau)\,d\tau
$. As $Y_i^{\prime}(t)=R_i(t)$, we consider initial value problems $Y_1'(t)=A_1^{\delta}(t)Y_1(t)$ and $ Y_2'(t)=A_2^{\delta}(t)Y_2(t)$ with initial conditions $Y_i(0)=1$.

\item Set the initial values $Y_{1,0}=Y_{2,0}=1$.

\item For $n=0,1,\ldots,N-1$, define $F_1^{\delta}(t,Y_1)=A_1^{\delta}(t)Y_1(t)$ and $F_2^{\delta}(t,Y_2)=A_2^{\delta}(t)Y_2(t)$.

For each $i=1,2$, calculate the RK4 stages
\[
K_{i,1}=F_i^{\delta}(t_n,Y_{i,n}),
\]
\[
K_{i,2}=F_i^{\delta}\left(t_n+\frac{h}{2},Y_{i,n}+\frac{h}{2}K_{i,1}\right),
\]
\[
K_{i,3}=F_i^{\delta}\left(t_n+\frac{h}{2},Y_{i,n}+\frac{h}{2}K_{i,2}\right),
\]
\[
K_{i,4}=F_i^{\delta}\left(t_n+h,Y_{i,n}+hK_{i,3}\right).
\]

\item Update $Y_{i,n+1}=Y_{i,n}+\frac{h}{6}\left(K_{i,1}+2K_{i,2}+2K_{i,3}+K_{i,4}
\right)$.

\item Recover the numerical approximations of $R_1^{\delta}$ and $R_2^{\delta}$: $R_{1,n}^{\delta}=A_1^{\delta}(t_n)Y_{1,n}$ and $R_{2,n}^{\delta}=A_2^{\delta}(t_n)Y_{2,n}.$

\item Approximate the derivatives at the interior nodes by the second-order central difference formula $(R_{i,n}^{\delta})^{\prime}
\approx
\frac{R_{i,n+1}^{\delta}-R_{i,n-1}^{\delta}}{2h},
\;
n=1,\ldots,N-1$. At the left endpoint use $R_{i,0}^{\delta}
\approx\frac{-3R_{i,0}^{\delta}+4R_{i,1}^{\delta}-R_{i,2}^{\delta}}{2h}$
and at the right endpoint use $(R_{i,N}^{\delta})^{\prime}
\approx\frac{3R_{i,N}^{\delta}-4R_{i,N-1}^{\delta}+R_{i,N-2}^{\delta}}{2h}$.

\item Calculate $P_{1,n}^{\delta}=-\frac{(R_{1,n}^{\delta})^{\prime}}{R_{1,n}^{\delta}},$ and $P_{2,n}^{\delta}=-\frac{(R_{2,n}^{\delta})^{\prime}}{R_{2,n}^{\delta}}.$

\item To stabilize the reconstruction of $P_i(t)$ for $i=1,2$, we add Tikhonov regularization. Since $P_i(t)=-\frac{d}{dt}\ln R_i(t)$, at first we regularize $g_i^{\delta}(t)=\ln R_i^{\delta}(t)$ by solving 
$$g_{i,\alpha}^{\delta}=\argmin_{\substack{g}}\{||g-g_i^{\delta}||_2^2+\alpha_i||D_2 g||_2^2\},$$
where $||D_2g||_2^2=\sum_{n=1}^{N-1}\left(\frac{g_{n-1}-2g_n+g_{n+1}}{h^2}\right)^2$, then calculate $P_{i,\alpha}^{\delta}(t)=-D_1g_{i,\alpha}^{\delta}(t)$ which is more stable than directly differentiating the oscillatory $R_i^{\delta}$, where $D_1$ differentiates the regularized logarithmic function and $\alpha_i>0$ is the Tikhonov regularization parameter selected automatically by generalized cross-validation $\alpha_i=\argmin\limits_{\alpha>0}\text{GCV}_i(\alpha)$ where $\text{GCV}_i(\alpha)=\frac{\frac{1}{N+1}||(I-H_{\alpha})g_i^{\delta}||_2^2}{\left(1-\frac{\text{tr}(H_{\alpha})}{N+1}\right)^2}$ and $H_{\alpha}=(I+\alpha D_2^{T}D_2)^{-1}$.
\end{enumerate}

\begin{figure}[ht]
    \centering
    \includegraphics[width=0.97\linewidth]{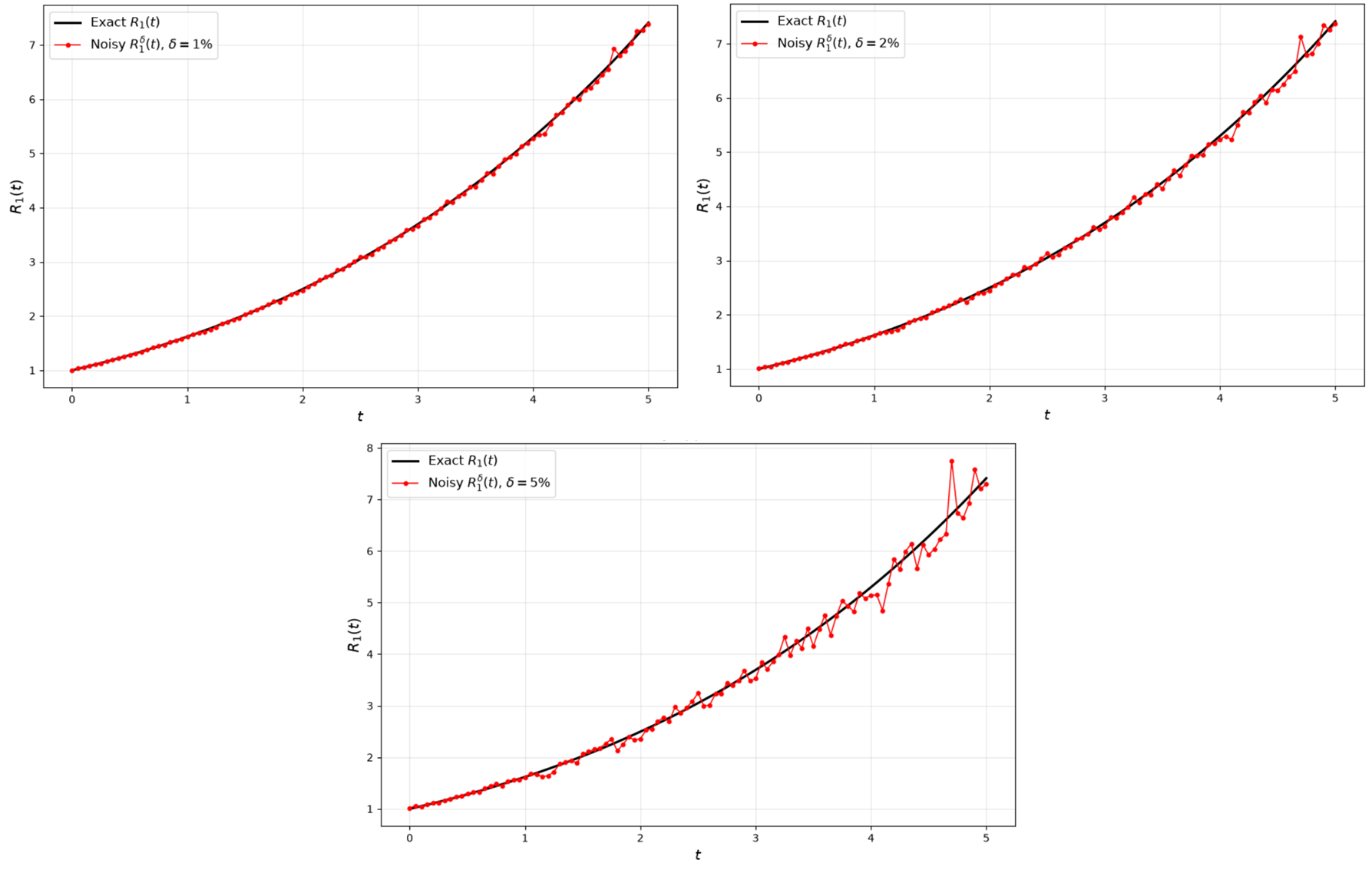}
    \caption{Reconstruction of unregulated $R_1(t)$ for different noise levels $\delta=1\%,2\%,5\%$ and $N=100$.}
    \label{fig:1}
\end{figure}
Figure \ref{fig:1} presents the reconstruction of the auxiliary coefficient $R_1(t)$ for noise levels $\delta=1\%,2\%,5\%$, using $N=100$ time steps. The noise is introduced through the Stefan boundary measurement, after which $R_1^{\delta}(t)$ is reconstructed through the auxiliary initial-value problem solved numerically by the RK4 scheme. The numerical curves show very good agreement with the exact solution. For $\delta=1\%$, the noisy reconstruction is almost indistinguishable from the exact curve, while only small deviations appear for $\delta=2\%$. At $\delta=5\%$, the perturbations become more visible, particularly toward larger values of $t$, where $R_1(t)$ itself has increased substantially. Nevertheless, the reconstruction retains the correct increasing behavior and remains concentrated around the exact solution.

Figure \ref{fig:2} gives the corresponding results for $R_2(t)$. Unlike $R_1(t)$, the exact $R_2(t)$ decreases over the considered time interval. Again, the approximations obtained with 1\% and 2\% noise follow the exact solution very closely, whereas the 5\% case displays more noticeable local fluctuations.
An important observation is that the reconstruction procedure performs well for coefficients having qualitatively different dynamics: $R_1(t)$ is increasing whereas $R_2(t)$ is decreasing. Hence, the method is not restricted to a particular monotonic behavior of the unknown temporal factor. Furthermore, the noisy $R_2^{\delta}(t)$ remains positive in the numerical experiment, which is important for the subsequent logarithmic representation used to recover $P_2(t)$.

\begin{figure}[ht]
    \centering
    \includegraphics[width=0.97\linewidth]{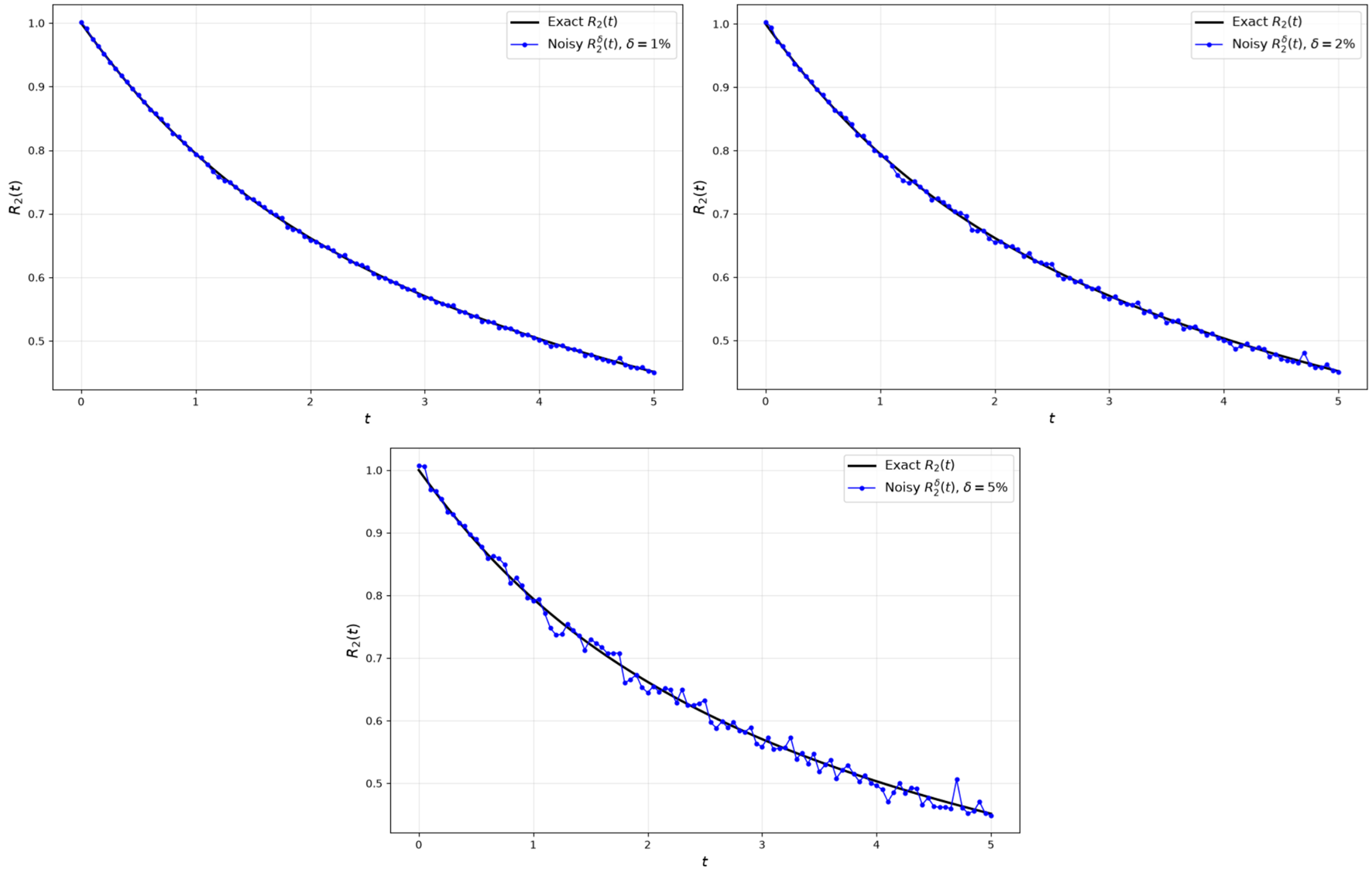}
    \caption{Reconstruction of unregulated $R_2(t)$ for different noise levels $\delta=1\%,2\%,5\%$ and $N=100$.}
    \label{fig:2}
\end{figure}

Figure \ref{fig:3} is particularly important because it reveals the principal numerical difficulty of the inverse reaction-coefficient problem. Although Figures \ref{fig:1} and \ref{fig:2} show that $R_1^{\delta}(t)$ and $R_2^{\delta}(t)$ can be reconstructed accurately, the directly reconstructed $P_1^{\delta}(t)$ and $P_2^{\delta}(t)$ exhibit strong oscillations even for the relatively small noise level $\delta=2\%$. Differentiation is well known to amplify high-frequency perturbations. Consequently, small pointwise fluctuations in $R_i^{\delta}(t)$, which are almost visually negligible in Figures \ref{fig:1} and \ref{fig:2}, become large fluctuations after differentiation. Figure \ref{fig:3} illustrates this effect very clearly: the noisy values fluctuate substantially above and below the exact curves, while the underlying exact $P_1(t)$ and $P_2(t)$ are smooth.
Therefore, Figure \ref{fig:3} identifies the main instability mechanism of the numerical inverse problem. The unstable behavior does not primarily originate from reconstruction of $R_i(t)$; rather, it is introduced during the differentiation required to recover $P_i(t)$. This observation provides the numerical motivation for introducing Tikhonov regularization.

\begin{figure}[ht]
    \centering
    \includegraphics[width=1.0\linewidth]{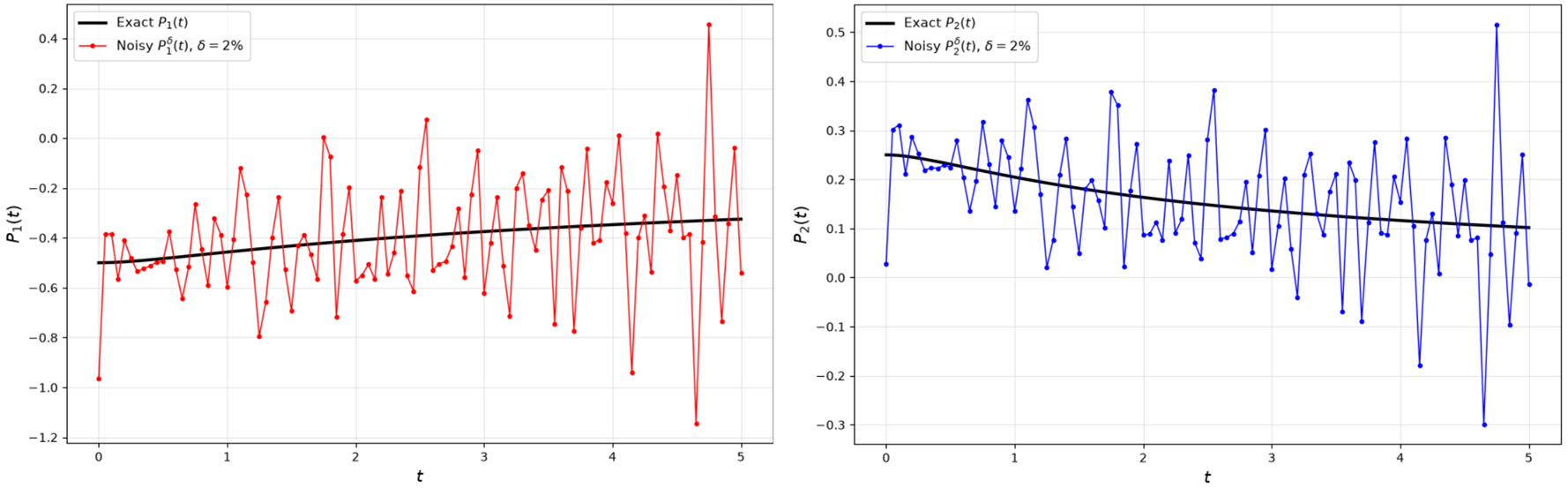}
    \caption{Reconstruction of unregulated $P_1(t)$ and $P_2(t)$ at noise level $\delta=2\%$ and $N=100$.}
    \label{fig:3}
\end{figure}

Figure \ref{fig:4} demonstrates the principal improvement of the proposed numerical reconstruction. Instead of directly differentiating $R_i^{\delta}(t)$, the method first introduces $g_i^{\delta}(t)=\ln R_i^{\delta}(t)$
and determines a regularized approximation by solving $g_{i,\alpha}^{\delta}=\argmin\limits_{g}\{||g-g_i^{\delta}||_2^2+\alpha_i||D_2g||_2^2\}$. The reaction coefficient is subsequently calculated as $P_{i,\alpha}^{\delta}(t)=-D_1 g_{i,\alpha}(t)$.
Thus, the second-derivative penalty suppresses rapid artificial oscillations before numerical differentiation is performed. The paper uses generalized cross-validation to determine $\alpha_i$ automatically. In contrast with Figure \ref{fig:3}, the reconstructed curves in Figure \ref{fig:4} are smooth and remain close to the exact coefficients for all three noise levels $\delta=1\%,2\%,5\%$.In particular, the regularized $P_1(t)$ preserves the increasing behavior of the exact negative coefficient, whereas $P_2(t)$ correctly preserves its positive decreasing behavior. Even at 5\% noise, the global temporal structure of both coefficients is recovered successfully. Regularization therefore transforms an unstable differentiation problem into a numerically stable reconstruction problem. Consequently, Figure \ref{fig:4} demonstrates not merely a visually smoother approximation but a substantial stabilization of the inverse reconstruction. The price of regularization is a small smoothing bias, most visible close to the endpoints, but this bias is considerably smaller than the large variance and oscillations observed without regularization.

\begin{figure}[ht]
    \centering
    \includegraphics[width=1.0\linewidth]{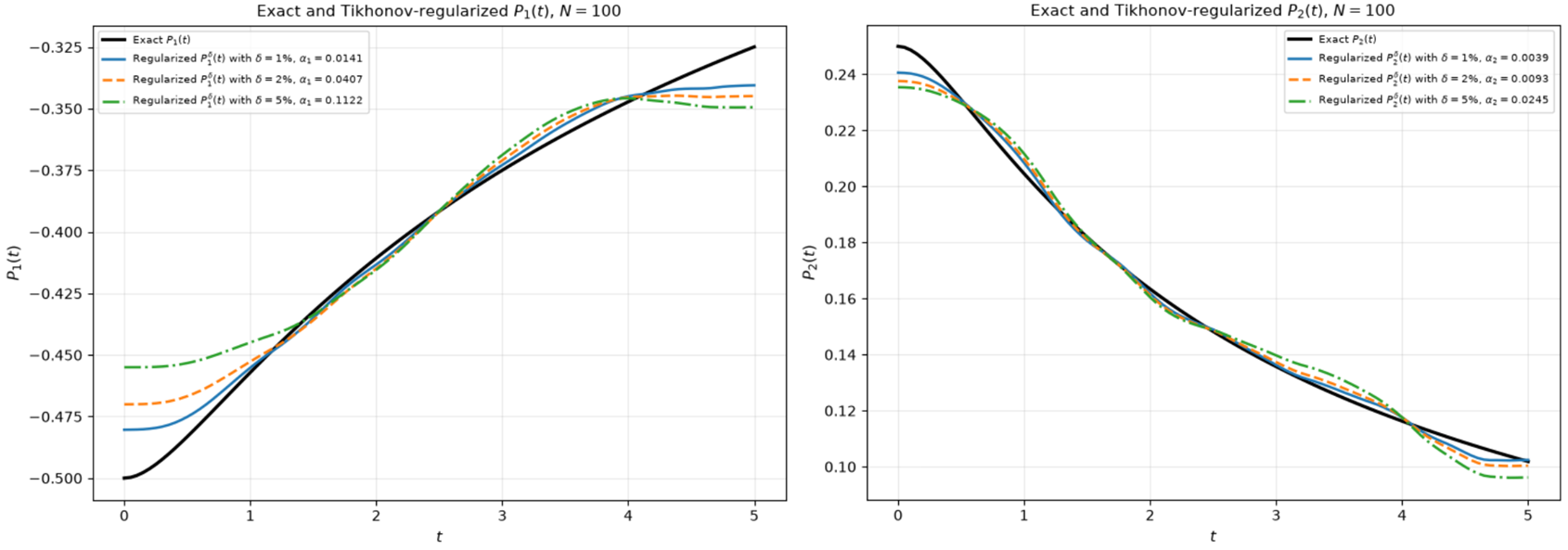}
    \caption{Reconstruction of regulated $P_1(t)$ and $P_2(t)$ at different noise levels $\delta=1\%, 2\%,5\%$ with $N=100$.}
    \label{fig:4}
\end{figure}

Figure \ref{fig:5} provides a quantitative assessment of the regularized reconstructions through the absolute errors $|P_{i,\alpha}^{\delta}(t)-P_i(t)|$ with $i=1,2$. For both coefficients, the error generally increases as the noise level changes from 1\% to 2\% and then to 5\%. However, the errors remain bounded throughout the entire time interval. From the displayed results, even for $\delta=5\%$, the largest error is approximately $4.5\times 10^{-2}$ for $P_1(t)$ and approximately $1.5\times 10^{-2}$ for $P_2(t)$. The larger errors close to the endpoints are also consistent with the numerical procedure. Derivative estimation is inherently less accurate near the boundaries, where centered finite differences cannot be applied in the same manner as at interior points. Moreover, smoothing methods frequently exhibit some boundary bias. The important result is therefore not that the reconstruction becomes completely independent of measurement noise, but rather that the effect of noise remains controlled and quantitatively bounded after regularization. Figure \ref{fig:5} consequently supplies direct numerical evidence of stability with respect to perturbations in the Stefan data. It also shows that $P_2(t)$ is reconstructed with a smaller absolute error than $P_1(t)$ for this particular numerical example, indicating that the two reaction coefficients may have different sensitivities to the same level of boundary perturbation.

\begin{figure}[ht]
    \centering
    \includegraphics[width=1.0\linewidth]{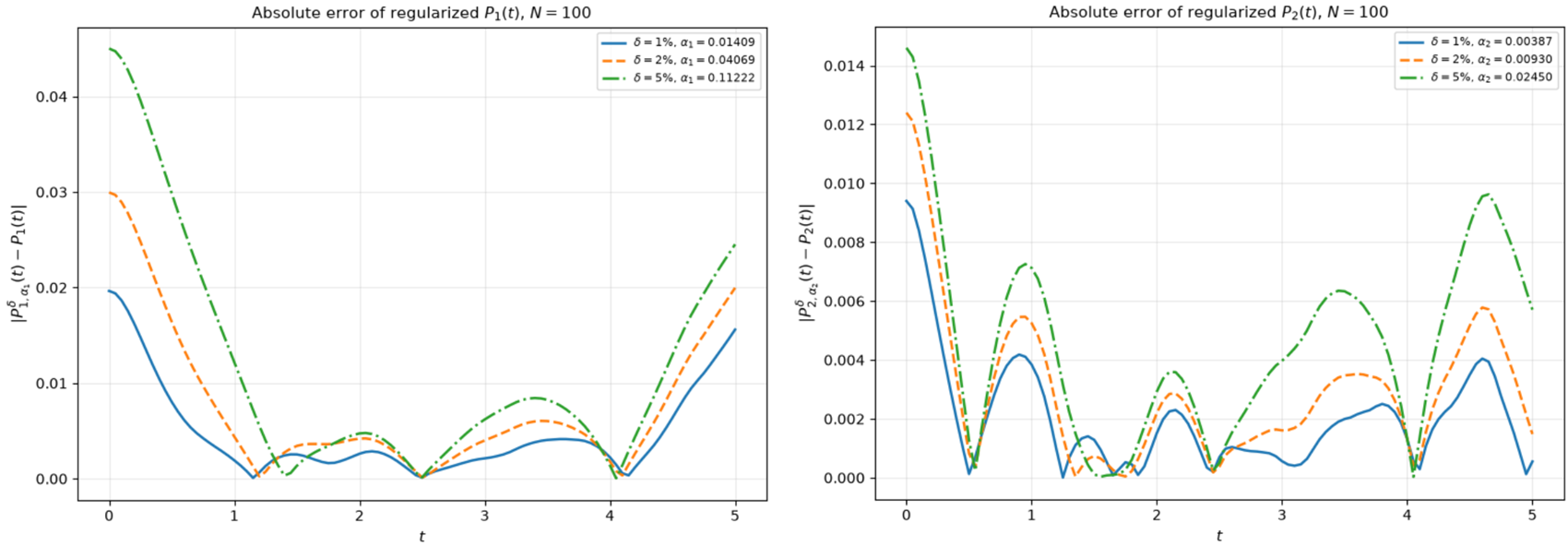}
    \caption{Absolute errors of regularized $P_1(t)$ and $P_2(t)$ at different noise levels $\delta=1\%, 2\%,5\%$ with $N=100$.}
    \label{fig:5}
\end{figure}

\begin{figure}[ht]
    \centering
    \includegraphics[width=1.0\linewidth]{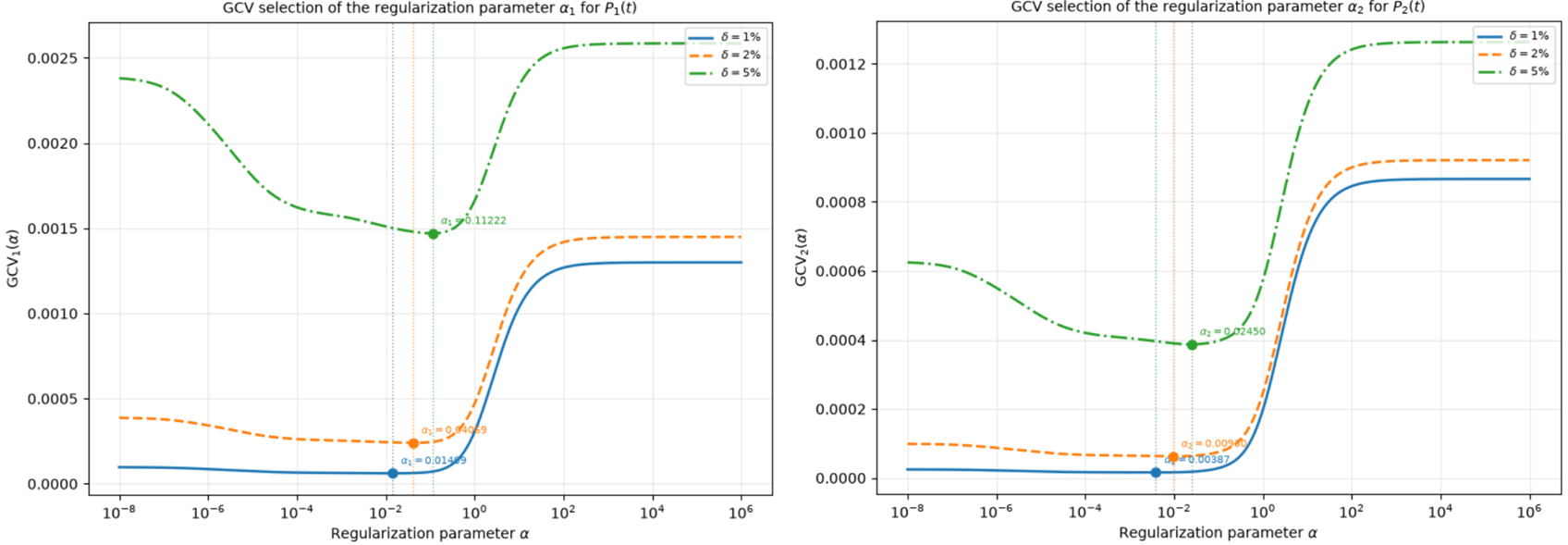}
    \caption{Generalized cross-validation (GCV) curves for selection $\alpha_1$ and $\alpha_2$  to reconstruct the regularized  $P_1(t)$ and $P_2(t)$ at different noise levels $\delta=1\%, 2\%,5\%$ with $N=100$.}
    \label{fig:6}
\end{figure}

Figure \ref{fig:6} demonstrates how generalized cross-validation determines the regularization parameters $\alpha_1$ and $\alpha_2$. The selected parameter corresponds to the minimum of each GCV curve. A clear and physically meaningful tendency is observed: the optimal regularization parameter increases with increasing noise level. More highly contaminated data therefore require stronger smoothing. This behavior is exactly what should be expected from a well-balanced regularization procedure. The fact that $\alpha_1$ and $\alpha_2$ are different is also significant. It shows that GCV adapts independently to the numerical sensitivity and scale of each phase instead of imposing the same arbitrary smoothing parameter on both coefficients. Figure \ref{fig:6} therefore adds an important methodological contribution: the regularization parameter does not need to be manually chosen. The GCV criterion provides a reproducible and data-driven mechanism for selecting the amount of regularization, reducing subjectivity and avoiding both under-regularization, which leaves noise in the solution, and over-regularization, which would excessively distort the true coefficient.

\section{Conclusion}
In this paper, we investigated inverse two-phase Stefan type problems for heat equations with time-dependent source and reaction coefficients. The main objective was to determine the temperature distributions in two phases together with unknown temporal coefficients appearing in the governing parabolic equations. The mathematical models were formulated in domains with moving boundaries, and suitable transformations were introduced to reduce the original free-boundary problems to equivalent problems on a fixed spatial interval.

For the first problem, the original two-phase Stefan problem was transformed into a system of parabolic equations with homogeneous boundary conditions on the fixed domain $(0,1)$. Using the eigenfunctions of the corresponding Sturm--Liouville problem, the weak solution was represented in the form of a Fourier series. Three types of additional information were considered: an integral overdetermination condition, a single observation point, and a non-local condition. For each case, reconstruction formulas for the unknown time-dependent source coefficients were obtained. These formulas were reduced to Volterra integral equations, which allowed us to establish the existence and uniqueness of the recovered coefficients under appropriate regularity, compatibility, and non-degeneracy assumptions.

Furthermore, the convergence of the spectral series was justified by using integration by parts, the Cauchy--Schwarz inequality, Bessel-type estimates, and Gronwall's inequality. These estimates ensured the boundedness of the reconstructed coefficients and the corresponding weak solutions. As a consequence, the existence and uniqueness of weak and strong solutions were established for the transformed problem, and therefore also for the original two-phase Stefan problem.

In the second formulation, we considered a two-phase Stefan type problem with Neumann boundary conditions and two moving boundaries. In this case, the unknown coefficients appeared as reaction terms $P_i(t)u_i(x,t)$. By introducing the transformation $R_i(t)=\exp\left(-\int_0^t P_i(\tau),d\tau\right),\; i=1,2,$ the original inverse reaction-coefficient problem was reduced to an inverse source-coefficient problem. The Stefan conditions at the moving boundaries provided the necessary relations for determining the auxiliary functions $R_1(t)$ and $R_2(t)$. Once these functions were recovered, the original reaction coefficients were uniquely reconstructed by $P_i(t)=-\frac{R_i'(t)}{R_i(t)}, \; i=1,2$. Thus, the second problem showed that the free-boundary conditions themselves can serve as sufficient information for recovering the unknown time-dependent coefficients, without imposing additional overdetermination conditions.

An illustrative example was also provided to demonstrate the applicability of the proposed approach. The example confirmed that the spectral method leads to explicit representations of the transformed temperature functions and to integral equations for the unknown coefficients.

The results of this work contribute to the theory of inverse Stefan problems by combining fixed-domain transformations, spectral decompositions, and Volterra integral equation techniques. The obtained existence, uniqueness, and regularity results provide a rigorous analytical basis for identifying time-dependent coefficients in phase-change heat transfer models. Future research may focus on numerical reconstruction algorithms, stability estimates with respect to noisy data, and extensions to nonlinear thermal coefficients, multidimensional geometries, and more general boundary conditions.

\textbf{Acknowledgments}

The research is supported by the Science Committee of the Ministry of Science and Higher Education of the Republic of Kazakhstan. (Grant No. AP32722175).

\textbf{Conflict of interest}

The authors declares no conflict of interest.

\textbf{Funding}

The research is supported by the Science Committee of the Ministry of Science and Higher Education of the Republic of Kazakhstan. (Grant No. AP32722175).

\textbf{Ethical approval}

This work don't need the approval from any ethics committees.

\textbf{Availability of data and materials}

Data sharing is not applicable to this article as no datasets were generated or analysed during the current study.

\textbf{Author Contributions Statement}

\textbf{Targyn Nauryz:} Conceptualization, Methodology, Supervision, Writing – original draft, writing - review and editing.

\end{document}